\documentclass[journal,onecolumn,draftclsnofoot,12pt]{IEEEtran}
\usepackage{bm}
\usepackage{amsmath,amssymb,amsfonts}
\usepackage{algorithmic}
\usepackage{graphicx}
\usepackage{textcomp}
\usepackage{xcolor}
\usepackage{changes}
\usepackage{cleveref}
\usepackage[backend=biber,style=ieee]{biblatex}
\usepackage{amsthm}
\usepackage{algorithm}
\usepackage{subfigure}
\usepackage{multirow}
\usepackage{verbatim}
\usepackage{url}
\usepackage{tikz}
\usetikzlibrary{matrix}
\usepackage{nicematrix}
\usepackage{mathtools}
\usepackage{stackengine}

\newtheorem{theorem}{Theorem}
\newtheorem{lemma}{Lemma}

\newtheorem{definition}{Definition}
\newtheorem{remark}{Remark}
\newtheorem{example}{Example}
\newtheorem{cor}{Corollary}

\newcommand{\Adn}{\mathbf{A}_{n,d}}

\newcommand{\Prob}{\mathbb{P}r}

\author{
\IEEEauthorblockN{Javad Maheri, K. K. Krishnan Namboodiri, and Petros Elia}\\
\IEEEauthorblockA{EURECOM Institute 
Sophia Antipolis, France\\
Email: \{maheri, karakkad, elia\}@eurecom.fr}
}
\date{}

\begin{document}

\title{Universal and Asymptotically Optimal Data and Task Allocation in Distributed Computing\thanks{This work was supported by the Huawei France-funded Chair towards Future Wireless Networks, by the French government under the France 2030 ANR program “PEPR Networks of the Future” (ref. ANR-22-PEFT-0010), and by European Research Council ERC-StG Project SENSIBILITÉ under Grant 101077361.}}

\maketitle

\begin{abstract}
    We study the joint minimization of communication and computation costs in distributed computing, where a master node coordinates $N$ workers to evaluate a function over a library of $n$ files. Assuming that the function is decomposed into an arbitrary subfunction set $\mathbf{X}$, with each subfunction depending on $d$ input files, renders our distributed computing problem into a \(d\)-uniform hypergraph edge partitioning problem wherein the edge set (subfunction set), defined by $d$-wise dependencies between vertices (files) must be partitioned across $N$ disjoint groups (workers). The aim is to design a file and subfunction allocation, corresponding to a partition \textcolor{black}{of \(\mathbf{X}\)}, that minimizes the communication cost $\pi_{\mathbf{X}}$, representing the maximum number of distinct files per server, while also minimizing \textcolor{black}{ the computation cost \(\delta_{\mathbf{X}}\) corresponding to a maximal worker subfunction load}. For a broad range of parameters, we propose a deterministic allocation solution, the \emph{Interweaved-Cliques (IC) design}, whose information-theoretic-inspired interweaved clique structure simultaneously achieves order-optimal communication and computation costs, for a large class of decompositions $\mathbf{X}$. This optimality is derived from our achievability and converse bounds, which reveal --- under reasonable assumptions on the density of $\mathbf{X}$ --- that the optimal scaling of the communication cost takes the form $n/N^{1/d}$, revealing that our design achieves the order-optimal \textit{partitioning gain} that scales as $N^{1/d}$, while also achieving an order-optimal computation cost.
    Interestingly, this order optimality is achieved in a deterministic manner, and very importantly, it is achieved blindly from $\mathbf{X}$, therefore enabling multiple desired functions to be computed without reshuffling files.

\end{abstract}


\section{Introduction}

The efficient allocation of computational and communication resources lies at the center of various settings of distributed computing, especially when it involves complex functions of large volumes of data. Such demanding distributed computing settings typically involve the evaluation, across multiple nodes, of a vast number of different subcomputations or subfunctions of large datasets. This in turn brings to the fore the challenge of efficiently distributing these tasks and data across the computing nodes in a way that minimizes computation as well as communication loads across workers~\cite{jiang2015survey,vavilapalli2013apache}. 

\paragraph*{Task and data assignment for reducing computation and communication costs} This distinct emphasis on dataset and task allocation mechanisms marks a clear departure from the classical paradigm of \textit{distributed source coding for function computation}~\cite{KoM,SlW,HaK,Wag,LPVKP,OrR,Doshi}, where data and task placement are typically fixed. This shift in emphasis towards optimizing the assignment of functions and datasets across servers, has the potential to better leverage the partial separability of computations, but is not without intricate combinatorial and information-theoretic design challenges in both achievable schemes and converse bounds, as we see in various recent publications.

Among these recent works, we have seen~\cite{WSJC,WSJC2} focusing on linearly separable functions, and exploring a distributed computing scenario with \( N \) servers, a single user, and multiple requested functions over multiple datasets. The main aim of these two works is to design communication schemes that reduce the total communication load, doing so under various data placements as well as various degrees of straggler resilience. Building on this theme,~\cite{KhE,KhE2} extend the problem to a multi-user setting without stragglers, where each server serves multiple users. The aim here is again to reduce  communication as well as computation costs. Taking a novel approach that involves covering-code constructions and tessellation-based tilings, these same works propose schemes for task assignment across servers and for communication between servers and users. Similarly,~\cite{NPME} formulate a related joint design problem under strict limits on the number of tasks per server, optimizing task distribution and transmissions that minimize the worst-case communication cost across all user demands.   \nocite{gholami2025hierarchicalgradientcodingoptimal,GhoWanJah+25,QiaGaoBol+24,Mal2024}

A parallel line of research, originating from the work on Coded MapReduce in~\cite{LMYA}, again focuses on the same broader challenge: designing data and task assignments to mitigate communication bottlenecks, now through coded exchanges. Subsequent variants address stragglers~\cite{CCW}, fixed assignments~\cite{WaW}, heterogeneous or asymmetric channels~\cite{BWW,PNR}, and diverse topologies~\cite{BrEl}, often borrowing clique-inspired tools from coded caching and broadcast networks~\cite{WSJTC,YaJ,MaT}. This line of work likewise aims to optimize distributed function computation through structured task placement, highlighting the role of dataset subpacketization in reducing server-to-server communication under rank-one (bottleneck) communication links.  

All the above lines of research often share the goal of designing --- under various settings and assumptions --- task and data assignment methods that reduce the communication and computation loads of distributed computing. This same theme is shared by our approach as well.

\paragraph*{Preliminary description of the considered distributed computing model} 

In this work, we consider a general distributed computing setting, which involves a master node that coordinates with \(N\) worker nodes in order to compute an arbitrary function. This desired function is assumed to take as inputs from a library of \(n\) datasets or files, and is also assumed to be decomposable into several subfunctions, with each subfunction depending on different file combinations. As expected, this setting entails computation and communication costs; workers must each compute multiple assigned subfunctions in order to collectively realize the target computation, and must each be sent sufficiently many files, enough for them to compute their subfunctions. 
The design objective is therefore to minimize the communication %delay
\textcolor{black}{cost} --- reflected here by the maximum number of files sent to any one worker --- as well as minimize the %computational delay,
\textcolor{black}{computation cost,} reflected here by the maximum number of subfunctions assigned to any one worker.

\begin{figure}
\begin{center}   
\includegraphics[width=14cm]{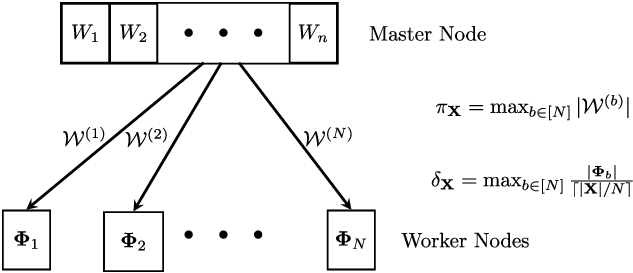}
\caption{Distributed computing model: The desired function \(F\) admits a decomposition \(F = \Psi(\{\zeta_\mathcal{T}(\mathcal{W}_\mathcal{T}):\mathcal{T}\in \mathbf{X},|\mathcal{T}|=d\}) \) for some \(\mathbf{X}\subseteq \mathbf{A}_{n,d}\). The set of files communicated to worker \(b\) is denoted with \(\mathcal{W}^{(b)}\), while \(\pi_{\mathbf{X}}\) denotes the worst-case communication cost across the master--worker links under the assumption of parallel, equal-capacity links. Similarly, the indices (\(d\)-tuples) of the subfunctions assigned for computation to worker \(b\) is \(\mathbf{\Phi}_{b}\), while \(\delta_{\mathbf{X}}\) denotes the computational delay normalized by the minimum possible computational delay, assuming homogeneous workers.   }
\label{SystemModel}
\end{center}
\end{figure}
We further assume that the master node has access 
to the library of \(n\) files \(\mathcal{W} = \{W_1, \dots, W_n\}\), and that each file \(W_j \in \mathbb{F}^B\) carries \(B\) symbols from some field \(\mathbb{F}\).  We also consider the case were the desired function, \(F: (\mathbb{F}^{B})^n\rightarrow \mathbb{F}^L\) is decomposed into subfunctions that each depends on \(d\) files. Consequently, any decomposition of $F$ can be represented by a subset \(\mathbf{X}\) of the set \(\mathbf{A}_{n,d}\) of all $d$-tuples from $1$ to $n$, so that our desired function \(F\) can be represented as 
\begin{equation} \label{generalFunction1}
\Psi(\{\zeta_\mathcal{T}(\mathcal{W}_\mathcal{T}):\mathcal{T}\in \mathbf{X},|\mathcal{T}|=d\}): (\mathbb{F}^T)^{|\mathbf{X}|} \rightarrow \mathbb{F}^L
\end{equation}
for some function $\Psi$, where each subfunction $\zeta_\mathcal{T}:(\mathbb{F}^B)^{d} \rightarrow \mathbb{F}^T$ operates on files $\mathcal{W}_\mathcal{T} = \{W_j, j\in \mathcal{T} \} \subset \mathcal{W}$. We often refer to the parameter \(d\) as the \emph{subfunction file degree} or simply \emph{degree}. 

As in most distributed computing settings, we have the following two phases.

\paragraph*{File and task allocation phase -- assigning subfunction set $\mathbf{\Phi}_b$ and file set $\mathcal{W}^{(b)}$ to each worker $b$} The master allocates to each worker $b$ a set $\mathbf{\Phi}_{b} \subset \mathbf{X}$ of subfunctions $\zeta_\mathcal{T}(\mathcal{W}_\mathcal{T}), \mathcal{T} \in \mathbf{\Phi}_{b}$, to compute. To do so, this worker needs to be sent all the file inputs to its assigned subfunctions. For $\mathcal{W}^{(b)} \subseteq \mathcal{W}$ being the set of files sent to worker $b$, the communication load on that link will equal $\pi_b = |\mathcal{W}^{(b)}|$ (in units of files), and thus the \emph{overall communication cost} 
\begin{equation}
\pi_\mathbf{X} = \max_{b\in \{1,\cdots,N\}} |\mathcal{W}^{(b)}|
\end{equation} will reflect the overall communication delay\footnote{At a time when multicasting is, unfortunately, rarely adopted in networks, this unicast-related metric captures well notions such as communication delay.}.
As one would expect, for any worker $b$ to compute subfunction $\zeta_\mathcal{T}(\mathcal{W}_\mathcal{T})$, the file and task allocations must guarantee\footnote{For example, let's say that $d=3$. If the file triplet $W_1,W_2,W_3$ does not appear at any one worker, then the subfunction $\zeta_{\{1,2,3\}}(\mathcal{W}_{\{1,2,3\}})$ cannot be computed.} that $\mathcal{W}_\mathcal{T} \subseteq \mathcal{W}^{(b)}$. Thus if  $\mathcal{I}_{\mathcal{W}^{(b)}}$ denotes the index set of files allocated to worker $b$, then the design must guarantee that 
\begin{equation}
\bigcup_{\mathcal{T}\in \mathbf{\Phi}_{b}} \mathcal{T} \subseteq \mathcal{I}_{\mathcal{W}^{(b)}}
\end{equation}
which again simply says that each worker must be communicated all necessary files ($\mathcal{W}^{(b)}$) for it to be able to compute its assigned subfunctions $\mathbf{\Phi}_{b}$. 

\paragraph*{Computing phase -- each worker $b$ computes subfunctions $\mathbf{\Phi}_{b}$} During this phase, each worker $b$ proceeds to compute all the subfunctions $\zeta_\mathcal{T}(\mathcal{W}_\mathcal{T})$ for all $\mathcal{T} \in \mathbf{\Phi}_{b}$. Assuming identical computing capability across workers, and identical computation cost per subfuction, the computation time becomes proportional to the maximum number of subfunctions assigned to any worker. Thus, for a given decomposition $\mathbf{X}$, minimizing this computation time, is equivalent to minimizing
\begin{equation}
\delta_{\mathbf{X}} = \frac{\max_{b\in [N]} |\mathbf{\Phi}_{b}|}{\lceil|\mathbf{X}|/N\rceil}
\end{equation}
which entails aiming for $\delta_{\mathbf{X}} $ that is as close to unity as possible. In the above, the denominator represents the ideal uniform load corresponding to the minimum possible computation delay. Our goal is to reduce the communication cost $\pi_{\mathbf{X}} $ and the computation cost $\delta_{\mathbf{X}} $.
\begin{example} Consider a function that takes inputs from a library of $n=7$ files, and consider a decomposition $\mathbf{X} = \{12,13,23,45,36,27\}\subset \mathbf{A}_{7,2}$, that tells us that there are $6$ subfunctions involved, where every subfunction takes as input $d=2$ files, with the first pair\footnote{We here omit the commas when describing the pairs. Clearly, when we say $12$ we are referring to pair $(1,2)$, and so on.} corresponding to a subfunction that takes as input files $1$ and $2$, the second pair to a subfunction that takes as input files $1$ and $3$, and so on. Consider $N=2$ worker nodes. 
If we assign subfunctions $\mathbf{\Phi}_{1} = \{12,13\}$ to worker $1$, and subfunctions $\mathbf{\Phi}_{2} = \{23,45,27,36\}$ to worker $2$, then worker $1$ will need to be sent files $\mathcal{W}^{(1)} = \{1,2,3\}$ and worker $2$ will need to be sent files $\mathcal{W}^{(2)} = \{2,3,4,5,6,7\}$.  This solution entails a communication cost of $\pi_{\mathbf{X}} = \max\{|\mathcal{W}^{(1)}|,|\mathcal{W}^{(2)}|\} = 6$ and a computation cost of $\delta_{\mathbf{X}} = \frac{max\{|\mathbf{\Phi}_{1}|,|\mathbf{\Phi}_{2}|\}}{\lceil|\mathbf{X}|/N\rceil} = \frac{4}{3}$. If on the other hand, we assign subfunctions $\mathbf{\Phi}_{1} = \{12,13,45\}$ to worker $1$, and $\mathbf{\Phi}_{2} = \{23,27,36\}$ to worker $2$, we would have an improved $\pi_{\mathbf{X}} = 5$ and an optimal $\delta_{\mathbf{X}}  = 1$. 
\end{example}

\begin{remark}\label{rem:Framework}
Our framework directly captures a broad class of functions that \textcolor{black}{decompose} into \(d\)-way dependencies over subsets of input variables. Such structures arise across diverse application domains. In statistics, the computation of covariance matrices involves pairwise dependencies (\(d=2\)) among variables~\cite{ledoit2004well}, while higher-order cumulant estimation extends this to \(d>2\) dependencies~\cite{comon1994independent}. In machine learning, kernel methods such as kernel PCA and spectral clustering require evaluating all pairwise kernel functions~\cite{scholkopf1998nonlinear}; contrastive and metric-learning objectives, including pairwise loss~\cite{chen2020simple} and triplet loss~\cite{schroff2015facenet}, similarly depend on \(d=2\) and \(d=3\) terms, respectively. Scientific computing tasks such as particle and molecular dynamics simulations exhibit analogous pairwise or higher-order structures~\cite{Dhaliwal2022RandomFeaturesInteratomic,rahimi2007random}. In genomics and bioinformatics, exhaustive SNP--SNP interaction analyses~\cite{li2015overview} and sequence comparison problems~\cite{lobo2008basic} are modeled through multiway dependencies.
Even modern learning architectures display similar \(d\)-wise structures: transformer-based attention mechanisms compute pairwise token dependencies~\cite{vaswani2017attention}, while large-scale similarity search frameworks~\cite{johnson2019billion} operate on tuple-based computations. In many of these cases, the task set \(\mathbf{X}\) can be very large, perhaps even occupying a non-trivial fraction of the full collection \(\mathbf{A}_{n,d}\), emphasizing the need for allocation schemes that remain efficient and balanced even when the computation involves a very large number of subfunctions. 

\end{remark}
% % % 
\paragraph*{Connection to hypergraph partitioning} It is easy to see that this current general distributed computing problem can be equivalently formulated as a hypergraph edge partitioning problem. For the same \(n,\ d,\ N,\) and \(\mathbf{X}\), we consider a \(d\)-uniform hypergraph \(\mathcal{H} = ([n], \mathbf{X})\), with vertex set $[n]\triangleq \{1,2,\cdots,n\}$ (the files), and hyperedge set $\mathbf{X}$ (the subfunctions), where the goal is to partition \(\mathbf{X}\) into \(N\) groups \(\mathbf{\Phi}_1, \mathbf{\Phi}_2, \dots, \mathbf{\Phi}_N\), so as to minimize  
\[
\pi_{\mathbf{X}}  = \max_{b \in [N]} |\alpha(\mathbf{\Phi}_b)|
\]
subject to the constraint  
\[
\frac{\max_{b \in [N]} |\mathbf{\Phi}_b|}{\lceil|\mathbf{X}|/N\rceil} \leq \delta_{\mathbf{X}} 
\]
where \(\delta_{\mathbf{X}}  \geq 1\) is a given real-valued constraint, and where \(\alpha(\mathbf{\Phi}_b)\) denotes the set of vertices from $[n]$ that are incident to the hyperedges in group \(\mathbf{\Phi}_b\). Hence, our objective of minimizing our communication %load
\textcolor{black}{cost} \(\pi_{\mathbf{X}}\), becomes that of minimizing the maximum number of vertices incident to any group, while also minimizing $\delta_{\mathbf{X}} $ which entails having groups of similar sizes. 

The vast literature on hypergraph edge partitioning focuses on similar problems, and  
a closely related objective that has been extensively studied in the literature, is the minimization of the \emph{Average Replication Factor (\(\mathrm{ARF}\))} of vertices \cite{Aykanat,HEPart,HQPart,SATPart,PartisHard,MultiLevelVLSI,Multihyp,HEComm}, where the \(\mathrm{ARF}\) is defined as
\begin{equation}
\label{eq:ARF}
  \mathrm{ARF} = \frac{1}{n} \sum_{v \in \mathcal{V}} \big|\{\, b \in [N] : v \in \alpha(\mathbf{\Phi}_b) \,\}\big|  
\end{equation}
to represent the \emph{average} number of groups in which a vertex appears\footnote{It is worth noting that by a simple counting argument, the \(\mathrm{ARF}\) can be equivalently expressed as 
\(
\mathrm{ARF} = \frac{1}{n} \sum_{b=1}^{N} |\alpha(\mathbf{\Phi}_b)|.
\)
Since our formulation guarantees that \(|\alpha(\mathbf{\Phi}_b)| \leq \pi_{\mathbf{X}}\) for all \(b \in [N]\), it immediately follows that
\(
\mathrm{ARF} \leq \pi_{\mathbf{X}} \frac{N}{n}.
\)
Hence, minimizing \(\pi_{\mathbf{X}}\) also minimizes an upper bound on the \(\mathrm{ARF}\), establishing a direct relationship between the two objectives.}.  

Whether focusing on the \(\mathrm{ARF}\) or not, most works consider algorithmic, search-based, designs, while much less is known in the form of theoretical insights or performance guarantees, let alone on the fundamental limits of the partitioning problem itself. Let us recall some of the state of art, starting with the crucial first point that the hypergraph edge partitioning problem, which seeks to minimize the average replication factor, is known to be NP-hard even for graphs (\(d = 2\))~\cite{Zhang2017GraphEdgePartitioning}. Additional works, motivated by the problem's computational complexity and broad applicability, can be found here~\cite{TrillionEdges,NSGA,Heterogeneous,Streaming,Dynamic}, again focusing on algorithmic developments, proposing heuristic or search-based partitioning strategies tailored to large-scale graphs.

In terms of theoretical insights, for graphs (i.e., for $d=2$), certain results offer analytical bounds on achievable \(\mathrm{ARF}\), where for example, the works in~\cite{Dynamic,TrillionEdges} reveal that their proposed algorithm guarantees an \(\mathrm{ARF}\) no greater than \((n + |\mathbf{X}| + N)/n\). Additionally, Li et al.~\cite{Zhang2017GraphEdgePartitioning} showed that, for any graph, it is possible to achieve an \(\mathrm{ARF}\) of \(O(\sqrt{N})\), with a guaranteed upper bound \(\mathrm{ARF} \le 2\sqrt{N} + \frac{N}{n}\). This \(O(\sqrt{N})\) scaling is tight for complete graphs, as shown recently in~\cite{BIS}, which provides explicit constructions attaining \(\mathrm{ARF} = O(\sqrt{N})\).  Furthermore, the work in~\cite{ProjectivePlane} employs finite projective geometry to partition the edges of graphs, achieving an \(\mathrm{ARF}\) between \( 1.5\sqrt{N} \) and \( 2\sqrt{N} \). This achievability result holds for graphs with \( N = q^2 + q + 1 \) vertices, where \( q \) is a prime power. This constraint arises from the limited existence of finite projective planes.

The problem of hypergraph partitioning has been in fact connected to various distributed computing settings \cite{maheri2025constructing}. In the important work in \cite{Aykanat}, the authors investigate the role of graph partitioning in parallel sparse matrix–vector multiplication. In this context, a hypergraph is constructed from the given matrix by representing each row as a vertex and each column as a hyperedge connecting all vertices corresponding to rows with nonzero entries in that column. The communication and computation cost minimization problem is then formulated as an optimal hypergraph partitioning task, typically approached through vertex partitioning of the constructed hypergraph.

Whether in the context of distributed computing or not, the problem of resource and task allocation via hypergraph partitioning represents a vast and mature field of scientific inquiry, substantiated by hundreds of research contributions spanning over four decades~\cite{Multihyp}. This extensive body of literature is motivated by a dual imperative. First, the problem possesses a profound mathematical depth, connecting NP-hard optimization challenges to complex combinatorial structures, Ramsey theory, and finite geometry \cite{Taskallocsurvey}. Second, and perhaps more significantly, the field is driven by an abundance of critical applications that demand such solutions. In addition to distributed computing, and the other here aforementioned applications, this hypergraphic partitioning approach finds direct applicability in 
emerging frontiers such as the parallel training of Large Language Models \cite{vaswani2017attention}.

\subsection{Brief Summary of Main Results} Returning to our general distributed computing setting, we recall that we consider a desired function that accepts an arbitrary decomposition defined by a set $\mathbf{X} \subseteq \mathbf{A}_{n,d}$ that lists all the subfunctions that must be computed by listing the corresponding $d$-tuple of inputs in each subfunction. We also recall that our objective is to distribute these subfunctions across the $N$ workers, thus having to partition $\mathbf{X}$ into $N$ groups, in a way that minimizes the communication cost \(\pi_{\mathbf{X}}\), in the presence of a reduced balance factor \(\delta_{\mathbf{X}}\) that reflects computational delay. We address the joint minimization of \(\pi_{\mathbf{X}}\) and \(\delta_{\mathbf{X}}\) via a constructive design that is based on the principle of interweaved cliques that will govern clique intersections in a carefully constructed subgraph of $\mathbf{A}_{n,d}$. This principle, common in information theoretic disciplines like coded caching and coded distributed computing, results in a novel structured approach to the hypergraph edge partitioning problem, which in turn yields an explicit design for file and subfunction allocation across the worker computing nodes.  
The design, which we term as the \emph{interweaved-cliques (IC) design}, is applicable for \textcolor{black}{all system parameters \((n, d, N)\)}, while the performance guarantees appear for a very broad range of parameters \((n, d, N)\), as elaborated in Theorem~\ref{thm:main}. 

In terms of this performance, the proposed design provides deterministic guarantees for any given task set \(\mathbf{X}\subseteq \mathbf{A}_{n,d}\). This design, together with a converse, now reveal that for any \(\mathbf{X}\) of size $|\mathbf{X}| = \varphi |\mathbf{A}_{n,d}|,$ any $N$ up to approximately $\sqrt{\binom{n}{d}}$ and any $\varphi $ as low as approximately $\frac{\ln n}{n^{d/2}}$ ($\varphi \in [0,1]$ plays the role of the normalized size of $\mathbf{X}$), the optimal communication cost $\pi_{\mathbf{X}}^\star $ satisfies 
\[\pi_{\mathbf{X}}^\star \in [ \frac{\varphi^{\frac{1}{d}} n}{N^{\frac{1}{d}}}, \frac{4e\,n}{N^{\frac{1}{d}}}].\] 
Furthermore, when \(\mathbf{X}\) is viewed as obtained from a \emph{random} thinning of \(\mathbf{A}_{n,d}\) --- that is, each \(d\)-tuple of $\mathbf{A}_{n,d}$ is included independently in \(\mathbf{X}\) with fixed probability \(\varphi\) --- the same construction additionally achieves a %computational delay
\textcolor{black}{computation cost} \(\delta_{\mathbf{X}} \le 5\) with probability \textcolor{black}{ at least \(1-\frac{1}{n}\)}. In the end, we now know that for any subfunction set \(\mathbf{X}\) of fixed (non-vanishing) normalized size $\varphi>0$, the optimal --- over all file and task assignment methods --- \textcolor{black}{optimal} communication cost scales as $\pi_{\mathbf{X}}^\star \asymp n/N^{\frac{1}{d}}$, and that with high probability, this is achieved with a computation delay of $\delta_{\mathbf{X}}\le5$. Thus, in the case of fixed $\varphi$ which corresponds to hypergraphs that are not unboundedly sparse, the resulting scaling laws suggest that the optimal gain scales as $N^{\frac{1}{d}}$, and this order-optimal gain is achieved by the IC design. Detailed concentration bounds, sampling thresholds, and proofs are provided in Section~\ref{sec: Random task}. 

Finally, the last $4$ corollaries, are in the context of the fact that a given desired function can admit multiple representations \(\mathbf{X}\), each leading to a distinct set of subfunctions and dependencies among the input files. While identifying favorable representations $\mathbf{X}$ is not the focus of this work, these Corollaries~\ref{cor: second},~\ref{cor: third},~\ref{cor: forth} offer some insight in that direction. 
%%%%%%%%%%%%%%%%%%%%%%%%%%%%%%%%%%%%%%%%%%%%%
\subsection*{Algorithmic Traits of Interest}
Our design enjoys two very attractive algorithmic properties. First, in terms of complexity, the IC design is deterministic, rather than search-based, and thus carries a minimal complexity load.  Secondly, in our design, the resulting allocation of files to workers (or equivalently of vertices to groups) is blind to \(\mathbf{X}\). This means that this file allocation is designed once and reused for any desired function with any decomposition \(\mathbf{X}\) (the library, as well as \(d\) remain fixed), naturally with the task assignment (hyperedge grouping) being obtained by restricting the precomputed groups to \(\mathbf{X}\).  
This independence offers a significant operational benefit in realistic settings where multiple functions, taking as input files from a common library of datasets, are computed simultaneously or sequentially. To illustrate this more clearly, consider a sequence of desired computed functions indexed by \(i\), where the \(i\)-th function is decomposed based on a set \(\mathbf{X}_i \subseteq \mathbf{A}_{n,d}\). Under the proposed IC design, and as long as $d$ remains fixed, no reshuffling of files across the \(N\) workers is required between successive rounds. For each instance \(\mathbf{X}_i\), the subset of subfunctions processed by each worker may vary, but the file placement remains fixed and valid for all tasks.  For any \(\mathbf{X}_i\), the file allocation will remain fixed, and it will correspond to an %asymptotically
order-optimal (for broad \(n,N\)) communication cost  (as long as the different normalized sizes $\varphi_i>0$ are non-vanishing); order-optimal over all file-and-task allocation algorithms, not only the `blind' ones. With high probability, the same file allocation will also yield an order-optimal %computational delay 
computation cost of $\delta_{\mathbf{X}_i}\le 5$.

The above traits are a product of the interweaved cliques approach. These IC designs are not new to the information-theory community; they appear prominently in problems such as coded caching~\cite{MaN}, where clique-based user side-information structures are effectively exploited to satisfy multiple user demands through coded transmissions over rank-one channels. In contrast, in this work, the notion of a clique arises in a fundamentally different context, serving as the combinatorial seed for constructing the worker–file–task allocation. Specifically, the larger universe \(\mathbf{X}\subseteq \mathbf{A}_{n,d}\) is generated from a smaller seed \(\mathbf{A}_{f,d}\), where \(f\) can be much smaller than \(n\). Similar formulations involving two complete-set structures have also appeared in other coding-theoretic and information-theoretic contexts~\cite{MKR,BrEl2,BrEl,8437333,NaR,PNR2} (see also \cite{EngEli2017,ZhaoBazcoElia2023}).

\paragraph*{A basic example}
Before providing a rigorous combinatorial representation of our distributed computing problem, let us offer a simple example, considering what might be thought of as the worst-case scenario, where we need to partition $\mathbf{X} = \mathbf{A}_{n,d}$. 

\begin{example}
\label{ex:6,2}
In an $N=3$ worker node setting where subfunctions take as input $d=2$ out of a total of $n=6$ files, our aim is to partition \[
\mathbf{X}=\mathbf{A}_{6,2}=\{\{1,2\}, \{1,3\}, \{1,4\}, \{1,5\}, \{1,6\}, \{2,3\},\{2,4\},\] \[\{2,5\},\{2,6\}, \{3,4\}, \{3,5\},\{3,6\},\{4,5\}, \{4,6\}, \{5,6\}\}\]
into $3$ sets \(\mathbf{\Phi}_1\), \(\mathbf{\Phi}_2\), and \(\mathbf{\Phi}_3\), describing the subfunctions associated to each worker node. Our goal is to have the number of digits appearing in each group be reduced. Let us first consider the lexicographic partition 
\[\mathbf{\Phi}_1=\{ \{1,2\},\{1,3\}, \{1,4\}, \{1,5\}, \{1,6\}\},\] \[\mathbf{\Phi}_2=\{\{2,3\},\{2,4\}, \{2,5\}, \{2,6\},\{3,4\} \}, \] \[\mathbf{\Phi}_3=\{ \{3,5\},\{3,6\},\{4,5\},\{4,6\},\{5,6\}\}\] where clearly the first group involves all $n=6$ digits from $1$ through $n$, and thus the first worker node must be assigned all $6$ files $W_1$ through $W_6$, thus bringing about a maximal communication cost of $\pi_{\mathbf{X}} = n = 6$. On the other hand, the partition 
\[\mathbf{\Phi}_1=\{ \{1,2\},\{1,3\}, \{1,4\}, \{2,3\}, \{2,4\}\},\] \[\mathbf{\Phi}_2=\{\{1,5\},\{1,6\}, \{2,5\}, \{2,6\},\{5,6\} \}, \] \[\mathbf{\Phi}_3=\{\{3,4\}, \{3,5\},\{3,6\},\{4,5\},\{4,6\}\}\] implies a reduced communication cost $\pi_{\mathbf{X}}= 4$ as each group only involves $4$ digits; digits $\{1,2,3,4\}$ in $\mathbf{\Phi}_1$, then digits $\{1,2,5,6\}$ in $\mathbf{\Phi}_2$, and digits 
$\{3,4,5,6\}$ in $\mathbf{\Phi}_3$. The second solution is preferable, and as we will see, optimal. 
\end{example}

\begin{remark}
It is tempting to perform partition simply by placing each $d$-tuple of $\mathbf{X}$ uniformly at random into one of the $N$ columns/groups. It is the case, though, that such a randomized approach would perform very poorly for our objective of reducing \(\pi_{\mathbf{X}}\). In essence, given that each vertex belongs to as many as \({n-1 \choose d-1}\) $d$-tuples, the probability that a vertex is \emph{not} covered by the edges landing in a specific column is vanishingly small, which entails $\pi_{\mathbf{X}}\asymp n$. 
This highlights the root difficulty: edges overlap heavily, and this overlap must be controlled while keeping column sizes balanced. Thus, minimizing $\pi_{\mathbf{X}}$ is in principle not something a naive random split will do well. 
 \end{remark}

\subsection{Paper Organization and Notation}
The remainder of this paper is organized as follows. Section~\ref{sec:prob stat} presents the problem formulation and objectives, while Section~\ref{sec: Main results} provides the main result in Theorem~\ref{thm:main} and its Corollaries %~\cref{cor: first,cor: second,cor: third,cor: forth} 
\ref{cor: first}--\ref{cor: forth} along with proofs. Section~\ref{sec: CliqueScheme} introduces the IC design, together with the necessary lemmas and propositions supporting the main theorem. Finally, the conclusion is given in Section~\ref{sec: conclu}, while the appendices include all additional proofs. 

\paragraph*{Notation}
We represent \(d\)-tuples using bold lowercase letters, such as \(\mathbf{a} = \{a_1, a_2, \dots, a_d\}\). Sets of \( d \)-tuples are denoted by bold uppercase letters, such as \(\mathbf{X}\). For any set \(\mathbf{X}\) of \(d\)-tuples, we let \(\mathbf{X}^{\mathrm{lex}} = (\mathbf{a}_1, \ldots, \mathbf{a}_{|\mathbf{X}|}) \)
denote the lexicographically ordered list of elements of \( \mathbf{X} \). As mentioned before, we use \([n]\) to denote the set \(\{1,2, \ldots, n\}.\) We also 
use \(\mathbf{A}_{n,d} = \binom{[n]}{d}\) to denote the collection of all
possible subsets of \([n]\) of size \(d\). For a random variable \(R\),
\(\mathbb{E}(R)\) denotes the expectation \textcolor{black}{and \(\Prob (R)\) denotes the probability of an instance of \(R\)}. Finally, we use \(f(n)\asymp g(n)\) when there is a constant \(c_1\), \(c_2\), and \(n_0\) such that \(\forall n \ge n_0\), \(c_1\cdot g(n)\le f(n)\le c_2\cdot g(n)\).
\section{Problem Statement and Setup}
\label{sec:prob stat}
The considered system comprises \(N\) worker nodes (servers) and a master node that coordinates the computation of a \emph{desired function of \(n\) input files}. This function can be decomposed into subfunctions as in~\eqref{generalFunction1}, 
where each subfunction depends on \(d\) files with indices from \([n]\), and where the set \(\mathbf{X} \subseteq \mathbf{A}_{n,d}\) represents the set of subfunctions that must be computed.\footnote{%
We will often speak of `files in $[n]$' (rather than files whose indices are in $[n]$), and we will often also refer to \(\mathbf{X}\) as the set of subfunctions, while remembering that it is the collection of index sets of files forming the input of each of these subfunctions.%
}

As suggested before, the master assigns subfunctions \(\mathbf{\Phi}_b \subseteq \mathbf{X}\) to each worker \(b \in [N]\), where these groups jointly form set \begin{equation}
\label{eq: partition P}
    \mathbb{P} = \{\mathbf{\Phi}_1, \mathbf{\Phi}_2, \dots, \mathbf{\Phi}_N\}
\end{equation}
that partitions $\mathbf{X}$, thus guaranteeing
\begin{equation}
\label{eq: partition condition of X}
\bigcup_{b \in [N]} \mathbf{\Phi}_b = \mathbf{X}, \quad 
\mathbf{\Phi}_b \cap \mathbf{\Phi}_{b'} = \emptyset, \ \forall b \neq b'
\end{equation}
ensuring that no subfunction is redundantly assigned to multiple workers.  

To compute its subfunctions \(\mathbf{\Phi}_b\), each worker \(b\) receives from the master node all the files whose indices\footnote{For example, if \(\mathbf{\Phi}_b = \{\{1,2,3\}, \{1,2,4\}\} \subset \mathbf{A}_{7,3}\), 
then \(\alpha(\mathbf{\Phi}_b) = \{1,2,3,4\} \subset [7]\). For brevity, we will often refer to $\alpha(\mathbf{\Phi}_b)$ as the files sent to worker $b$.} are in set 
\begin{equation}
    \label{eq: alpha phi_b}
    \alpha(\mathbf{\Phi}_b) = \{i \in [n] \mid \exists \mathbf{a} \in \mathbf{\Phi}_b \ \text{with} \ i \in \mathbf{a}\}.
\end{equation}
Naturally, \(|\alpha(\mathbf{\Phi}_b)|\) represents the number of distinct files that must be communicated  from the master node to worker \(b\), and thus captures the communication load on that link.

Given \(n,\ d,\ N\), and the subfunction set \(\mathbf{X}\), 
we let
\(\mathfrak{S}_{\mathbf{X}}(n,d,N)\) denote the class of all valid allocation schemes that partition 
\(\mathbf{X}\) into \(N\) disjoint groups \(\mathbf{\Phi}_1, \dots, \mathbf{\Phi}_N\) as in \eqref{eq: partition P} and \eqref{eq: partition condition of X}. 
Any scheme \(\boldsymbol{\mathcal{S}} \in \mathfrak{S}_{\mathbf{X}}(n,d,N)\), 
entails a \emph{communication cost}
\begin{equation}
    \label{eq: com cost }
\pi^{\boldsymbol{\mathcal{S}}}_{\mathbf{X}} = \max_{b \in [N]} |\alpha(\mathbf{\Phi}_b)|
\end{equation}
% factor parameter
and a \emph{computation cost}
\begin{equation}
    \label{eq: compt cost}
    \delta^{\boldsymbol{\mathcal{S}}}_{\mathbf{X}} = \frac{\max_{b \in [N]}|\mathbf{\Phi}_b|}{\lceil |\mathbf{X}|/N \rceil}
\end{equation}
where these are bounded as 
\[\pi^{\boldsymbol{\mathcal{S}}}_{\mathbf{X}} \in [d,n], \ \ \delta^{\boldsymbol{\mathcal{S}}}_{\mathbf{X}} \in [1,N]\]
with $\pi^{\boldsymbol{\mathcal{S}}}_{\mathbf{X}} = n$ corresponding to a scheme where one server is allocated all the files, and with $\delta^{\boldsymbol{\mathcal{S}}}_{\mathbf{X}} = N$ corresponding to a solution where one server is assigned all subfunctions. Any scheme will entail a \emph{partitioning gain}, in reference to $n/\pi^{\boldsymbol{\mathcal{S}}}_{\mathbf{X}}$.

Finally, the corresponding optimal values over all valid schemes, are defined as
\[
\pi^\star_{\mathbf{X}} = \min_{\boldsymbol{\mathcal{S}}\in \mathfrak{S}_{\mathbf{X}}(n,d,N)} 
\pi^{\boldsymbol{\mathcal{S}}}_{\mathbf{X}}, 
\quad 
\delta^\star_{\mathbf{X}} = \min_{\boldsymbol{\mathcal{S}}\in \mathfrak{S}_{\mathbf{X}}(n,d,N)} 
\delta^{\boldsymbol{\mathcal{S}}}_{\mathbf{X}}
\]
and our objective is to construct an allocation scheme \(\boldsymbol{\mathcal{S}}\) for our distributed computing problem, that minimizes both \(\pi^{\boldsymbol{\mathcal{S}}}_{\mathbf{X}}\) and \(\delta^{\boldsymbol{\mathcal{S}}}_{\mathbf{X}}\), 
thereby minimizing the overall communication cost and computation %delay
\textcolor{black}{cost} in the system.

\section{Main Results}
\label{sec: Main results}

\begin{figure}
  \includegraphics[width=\linewidth]{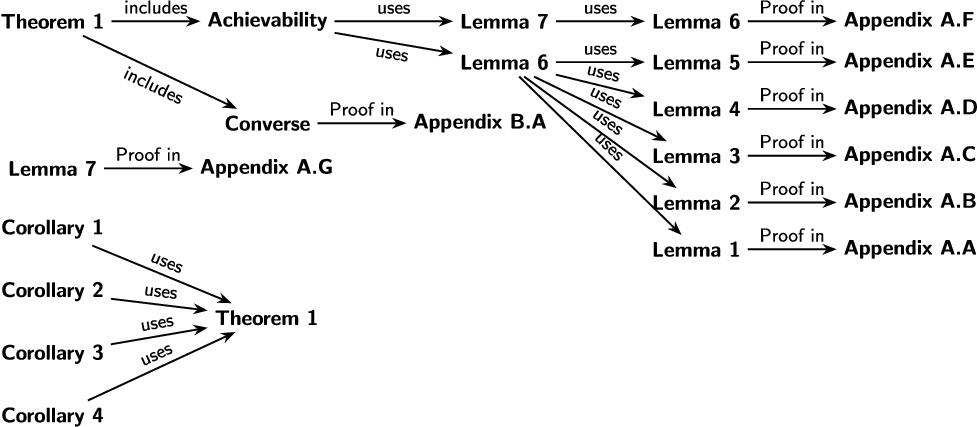}
  \caption{\textcolor{black}{Overview of the logical structure of the proofs of the paper, highlighting the dependencies among lemmas and theorems.}}
  \label{fig:boat1}
\end{figure}

We now present the main results of our distributed computing setting, providing lower and upper bounds on the optimal communication cost $\pi^\star_{\mathbf{X}}$ and computation cost $\delta^\star_{\mathbf{X}}$.
As suggested, any probabilistic statement involves viewing $\mathbf{X}$ as a realization of an independent random thinning of $\mathbf{A}_{n,d}$, in which each hyperedge is retained with probability $\varphi$. We proceed with the main theorem, which will hold for all 
\begin{equation} \label{eq:phi1}
\varphi \geq \varphi_{\mathrm{min}} = \frac{96\,N\,\log(2Nn)}{\binom{n}{d} - 2^{\,d+2}N} \approx \frac{ \ln n}{n^{d/2}}
\end{equation}
(see also later \eqref{eq:varphi_min}). Note also that $\varphi_{\mathrm{min}}$ is also nicely upper bounded by 
$C_d\cdot\frac{ \ln n}{n^{d/2}}$, where \(C_d=192 \cdot (e\cdot d)^{d/2}\) and where \(e\) is Euler's number.

\begin{theorem}
\label{thm:main}
   For the distributed computing setting with \(n\) files, \(N\) workers, subfunction degree \(d\), and subfunction set \(\mathbf{X} \subseteq \mathbf{A}_{n,d}\) with normalized size \(\varphi = |\mathbf{X}| / \binom{n}{d}\), the optimal communication cost $\pi^\star_{\mathbf{X}}$ satisfies
\[
\pi^\star_{\mathbf{X}} \in \left[\,\frac{\varphi^{\frac{1}{d}} n}{N^{\frac{1}{d}}},\; \frac{4e n}{N^{\frac{1}{d}}}\,\right].
\]
 
Furthermore, if \(\varphi \ge  \varphi_{\mathrm{min}}\), \(d \le \frac{n}{32}\) and \(N \le(\frac{9}{10}\sqrt{\frac{n}{d}})^{d}\), the above upper bound is guaranteed with \(\delta_{\mathbf{X}} \le 5\), with probability at least \(1 - \frac{1}{n}\), and these can be achieved by the IC design. 
\end{theorem}

\begin{proof}
The above is a direct outcome of the converse \(\pi_\mathbf{X}^\star\geq \frac{\varphi^{\frac{1}{d}} n}{N^{\frac{1}{d}}}\) in Appendix~\ref{proofoflem:fundlower} and of the achievability part \(\pi_\mathbf{X}^\star\leq \frac{4e n}{N^{\frac{1}{d}}}\), where achievability is due to the interweaved clique design from Section~\ref{sec: CliqueScheme}. The guarantee on the computational cost $\delta_{\mathbf{X}}$ is derived in Lemma~\ref{lem: conc on group size}, \textcolor{black}{in conjunction with the findings of the preceding Lemmas \ref{lem:ub-lb}-\ref{lem: delta pi, A,n,d}}.
\end{proof}

\begin{remark}
In Theorem~\ref{thm:main}, the parameter \( \varphi \) arises both as the normalized size of a given subfunction set \( \mathbf{X} \) and as a parameter governing the random {construction of \(\mathbf{X}\)}. On the one hand, for a given subfunction set \( \mathbf{X} \subseteq
\mathbf{A}_{n,d} \) with normalized size \( \varphi = |\mathbf{X}|/\binom{n}{d} \), the communication cost \( \pi_{\mathbf{X}} \le \frac{4e n}{N^{\frac{1}{d}}} \) is achievable. On the other hand, {\( \varphi \) also appears as the sampling probability used to construct \( \mathbf{X} \) by independently sampling elements of \( \mathbf{A}_{n,d} \)}. Under the stated conditions in Theorem~\ref{thm:main}, {for any \( \mathbf{X} \) obtained from this sampling procedure,}, the proposed IC design achieves the same communication cost \( \pi_{\mathbf{X}} \le \frac{4e n}{N^{\frac{1}{d}}} \); moreover, with probability at least \(1-\tfrac{1}{n}\), it also ensures that \(\delta_{\mathbf{X}} \le 5\).
\end{remark}

\begin{remark}
The above theorem is not of an asymptotic nature. In the large $n$ regime though, the theorem further simplifies to say that if \(\varphi \geq \varphi_\mathrm{min} \approx \frac{ \ln n}{n^{d/2}}\), and  \(N \le (\frac{9}{10}\sqrt{\frac{n}{d}})^{d}\approx \sqrt{\binom{n}{d}}\), the additional guarantee of \(\delta_{\mathbf{X}} \le 5\), is achieved with probability almost one. %at least \((1 - \frac{1}{n})\rightarrow 1\). 
\end{remark}
The following corollary distills the above theorem in the form of scaling laws. We will henceforth use the term \textit{$(n,N,d,\varphi)$-distributed computing setting}, to refer to our setting with \(n\) files, \(N\) workers, subfunction degree \(d\), and a subfunction set \(\mathbf{X} \subseteq \mathbf{A}_{n,d}\) of normalized size \(\varphi = |\mathbf{X}| / \binom{n}{d}\).  We finally recall that the term \emph{partitioning gain} refers to $n/\pi$.
\begin{cor}
\label{cor: first}
% d<n/4, for delta
In our $(n,N,d,\varphi)$-distributed computing setting, with %asymptotically large 
scaling $n,N \le(\frac{9}{10}\sqrt{\frac{n}{d}})^{d} $ and fixed \(d,\varphi\), the optimal communication cost scales as \(\frac{n}{N^{1/d}}\), yielding an order-optimal partitioning gain that scales as \(N^{\frac{1}{d}}\). This order-optimal performance is achieved by the interweaved clique design for any $\mathbf{X}$, and when \(N \le (\frac{9}{10}\sqrt{\frac{n}{d}})^{d}\), the same design also guarantees the order optimal \(\delta_{\mathbf{X}} \le 5\) with probability {at least \(1-\frac{1}{n} \rightarrow 1\)}. Thus order-optimality in $\pi_{\mathbf{X}} $ and $\delta_{\mathbf{X}}$ can be achieved simultaneously.
\end{cor}
\begin{proof}
    The result follows directly from Theorem~\ref{thm:main}.
\end{proof}

The preceding result establishes the optimal scaling of the communication cost for any \textit{fixed} subfunction set \(\mathbf{X}\). %, and the optimal scaling of the computational delay with high probability. 
As suggested before, though, a function \(F\) generally admits multiple valid decompositions of the form of~\eqref{generalFunction1}, each corresponding to a distinct subfunction set \(\mathbf{X}\). 
Given any function $F$, the following corollary readily bounds the optimal communication cost, jointly optimized over all subfunction decompositions $\mathbf{X}$ and all data-and-task assignment solutions for each valid decomposition.

\begin{cor}
\label{cor: second}
For any  function \(F\) with \(n\) input files, let \(\mathcal{X}_{F,d}\) be the set of all its possible decompositions \(\mathbf{X} \subseteq \mathbf{A}_{n,d}\) with fixed subfunction degree \(d\), and let \(\varphi_{\mathrm{x-min}} \triangleq \min_{\mathbf{X} \in \mathcal{X}_{F,d}} \{|\mathbf{X}| / \binom{n}{d} \}\). Then the optimal communication cost $\pi_F^{\star\star}$ over all decompositions \(\mathbf{X}\) and all data-and-task allocation policies, lies in the region
\[
\pi_F^{\star\star}
\in
\Big[\,\varphi_{\mathrm{x-min}}^{1/d} \cdot \tfrac{n}{N^{1/d}},\;
\tfrac{4e\,n}{N^{1/d}}\,\Big].
\]
Thus, %in the limit of large $n,N,$
for scaling \(n\) and \(N\le (\frac{9}{10}\sqrt{\frac{n}{d}})^{d}\) and under any non-vanishing $\varphi_{\mathrm{x-min}}>0$, this optimal cost scales as 
\[
\pi_F^{\star\star} \asymp \tfrac{n}{N^{1/d}}
\]
and thus the optimal partitioning gain scales as \(N^{1/d}\). 
Moreover, the proposed interweaved clique design achieves this order optimal gain, for every decomposition \(\mathbf{X}\in\mathcal{X}_F\).
\end{cor}

\begin{proof}
The proof is again direct from Theorem~\ref{thm:main}.
\end{proof}
The following considers the case where the decompositions span multiple values of $d$.  The corollary provides an optimization approach, based on bounds. 
\begin{cor}
\label{cor: third}
Let the function \(F\) admit a collection of admissible decompositions
\(
\mathcal{X}_F=\{\,\mathbf{X}_i \subseteq \mathbf{A}_{n,d_i} \,\},
\)
where each \(\mathbf{X}_i\) has their own degree \(d_i\) and density
\(
\varphi_{i} = |\mathbf{X}_i|/\binom{n}{d_{i}}. 
\)
Let $d_{\min}$ be the minimum subfunction degree $d$ in $\mathcal{X}_F$. 
Then the optimal communication cost, among all decompositions in \(\mathcal{X}_F\) and all data-and-task allocation schemes, lies in the region
\[
\bigl[ \min_{\,\mathbf{X}_i \in \mathcal{X}_F}\;
\frac{\varphi_{i}^{1/d_{i}}\, n}{N^{1/d_{i}}},  \ \ \frac{4en}{N^{1/d_{\min}}} \bigr].
\]
Furthermore, for $ i'$ such that $\mathbf{X}_{i'}$ has degree $d_{\min}$, for %sufficiently large \(n\),
\textcolor{black}{any given \(n\)}, if the following three conditions 
\textcolor{black}{\( \varphi_{_i'} \ge C_{d_{\min}}\cdot\frac{\ln n}{n^{{d_{\min}}/2}}\)}, \(d_{\min} \le \frac{n}{32}\) and \textcolor{black}{\(N \le (\frac{9}{10}\sqrt{\frac{n}{d_{\min}}})^{d_{\min}}\)} hold, then the above upper bound is guaranteed with \(\delta_{\mathbf{X}_{i'}} \le 5\), with probability at least \(1 - \frac{1}{n}\), and these can be achieved by the IC design.
\end{cor}
\begin{proof}
The proof is direct from Theorem~\ref{thm:main}.
\end{proof}
Finally the following considers the scaling laws of the previous corollary.
\begin{cor}
\label{cor: forth}
In our distributed computing setting with a desired function \(F\) 
with admissible decompositions
\(
\mathcal{X}_F=\{\,\mathbf{X}_i \subseteq \mathbf{A}_{n,d_i} \,\},
\) with non-vanishing \( \varphi_{i} = |\mathbf{X}_i^{(d_i)}| / \binom{n}{d_i}>0 \) and with $d_{\min}$ being the minimum subfunction degree $d$ in $\mathcal{X}_F$, then %in the limit of large 
for scaling \textcolor{black}{$n,N\le (\frac{9}{10}\sqrt{\frac{n}{d_{\min}}})^{d_{\min}}$}, the optimal communication cost, among all decompositions in \(\mathcal{X}_F\) and all data-and-task allocation schemes, scales as 
\[
\pi^{\star\star} \asymp \frac{n}{N^{1/d_{\min}}}
\]
and the optimal partitioning gain scales as \(N^{1/d_{\min}}\). 
\end{cor}
\begin{proof}
    The proof is direct from Theorem~\ref{thm:main}, and it again involves achievability from the IC design.
\end{proof}

To conclude, Theorem~\ref{thm:main} reveals that for any function, any subfunction decomposition \(\mathbf{X}\) with \(|\mathbf{X}| \geq \varphi |\Adn| \) will accept an optimal communication cost (measured, as we do so here, in the form of the maximum number of files $\pi_{\mathbf{X}}$ at any worker) that lies in the region $\pi_{\mathbf{X}}^\star \in [ \frac{\varphi^{\frac{1}{d}} n}{N^{\frac{1}{d}}}, \frac{4e\,n}{N^{\frac{1}{d}}}].$  There is no function, nor decomposition, whose optimal communication cost is outside the above region.  
This fact also captures potential efforts to decompose a desired function into a collection of \emph{non-atomic} subfunctions, taking as inputs smaller-in-size subfiles. Such subpacketization approach would tend to create a new larger $n$, new smaller files, and a new, possibly larger $\mathbf{X}$ that would again follow the rules of the fundamental limits described above, which in turn would dictate whether or not subpacketization gains would be possible.

In the end, under our metric, we see that the new IC design is a universal design, in the sense that --- under the above assumption of a fixed $\varphi$ --- it provides the order-optimal gain $N^{1/d}$, for any function and any decomposition.

\section{Achievable Scheme: Interweaved Clique (IC) Design}
\label{sec: CliqueScheme}
In this section, we introduce the IC design for data and task (subfunction) allocation in distributed computing. The design will partition any 
%'good' 
subset \(\mathbf{X}\subseteq\mathbf{A}_{n,d}\) into \(N\) groups \(\mathbf{\Phi}_{1}\), \(\mathbf{\Phi}_{2}\), \(\ldots\), \(\mathbf{\Phi}_{N}\), aiming to minimize $\pi_{\mathbf{X}}$ and $\delta_{\mathbf{X}}$, and it will involve three main steps. First, we will design a partition of the entire \(\mathbf{A}_{n,d}\), where this partition only involves \(N'\) groups. In this first step --- which will be described in Sections~\ref{subsec: Case1} and \ref{subsec: Case2} --- this $N'$ will take the form
\begin{equation}
\label{eq:N'}
    N' = \binom{k}{d}
\end{equation}
where \(k\) is defined as
\begin{equation}
\label{eq:largest_r}
    k \triangleq \max \left\{ r \in \mathbb{Z}^{+} \,\mid\, \binom{r}{d} \leq N \right\}.
\end{equation}
In the subsequent step, described in Section~\ref{subsec: N' to N}, we will introduce a technique to extend the partition, from $N'$ groups, to \(N\geq N'\) groups \(\widetilde{\mathbf{\Phi}}_1\), \(\widetilde{\mathbf{\Phi}}_2\), \(\ldots\), \(\widetilde{\mathbf{\Phi}}_{N}\), while the final step (Section~\ref{sec: Random task}), will mainly involve the creation of the final partition
\begin{equation}
\label{eq: partition refinement}
    \mathbf{\Phi}_{1}=\widetilde{\mathbf{\Phi}}_1\cap\mathbf{X}, \ \mathbf{\Phi}_{2} =\widetilde{\mathbf{\Phi}}_2\cap\mathbf{X},\ \ldots,\ \mathbf{\Phi}_{N}=\widetilde{\mathbf{\Phi}}_N\cap \mathbf{X}
\end{equation}
all while keeping the two costs in check.  When referring to the IC design, we will often use --- for brevity ---   \(\pi_{\mathbf{X}}\) and \(\delta_{\mathbf{X}}\), instead of \(\pi^{\boldsymbol{\mathcal{S}}}_{\mathbf{X}}\) and \(\delta^{\boldsymbol{\mathcal{S}}}_{\mathbf{X}}\). Furthermore, when  context allows, especially when we are partitioning \(\mathbf{X} = \mathbf{A}_{n,d}\), we may simply write \(\pi\) and \(\delta\), respectively, for brevity.

Let us start with the first step, and let us consider two separate cases: {Case~1} for when $n$ is divisible by the \(k\) defined in \eqref{eq:largest_r}, and {Case~2} for when $n$ is not divisible by \(k\).

\subsection{Case 1: \(k=\frac{n}{s}\) for some \(s\in \mathbb{Z}^+\)}
\label{subsec: Case1}
We here consider the case where $k$ from \eqref{eq:largest_r} divides $n$. Let $s =  \frac{n}{k}$, and let \(N'=\binom{k}{d}={n/s \choose d}.\) 
Let us first partition the file set \([n]\) into \(f\) disjoint families \(\mathcal{F}_1,\mathcal{F}_2,\dots,\mathcal{F}_f\), each taking the form
\begin{equation} \label{}
    \mathcal{F}_i\triangleq\{(i-1)s+1,(i-1)s+2,\dots,is\}, \ i\in[f]
\end{equation}
where we set \(f =  k=\frac{n}{s}\), 
and let us, for any subset \(\mathcal{I} \subseteq [f]\), define now --- to be used later --- the union of families 
\begin{equation}
\label{eq:F_I}
    \mathcal{F}_{\mathcal{I}} \triangleq \bigcup_{i \in \mathcal{I}} \mathcal{F}_i.
\end{equation}
Our aim now will be to design a partition 
\begin{equation}
\label{eq: A_n,d partition}
    \mathbf{A}_{n,d} = \bigcup\limits_{\sigma \in \binom{[f]}{d}} \widetilde{\mathbf{\Phi}}_\sigma
\end{equation}
where we partition $\mathbf{A}_{n,d}$ into $\binom{f}{d}$ disjoint groups $\widetilde{\mathbf{\Phi}}_\sigma$, each labeled by a $d$-tuple $\sigma \subset [f]$.
Before proceeding with the construction of this partition, we introduce the following definition of the 
\emph{support family} of a \(d\)-tuple.
\begin{definition}
For a \(d\)-tuple \(\mathbf{a} \in \mathbf{A}_{n,d}\), its \emph{support family} is defined as
\begin{equation*}
    \mathcal{B}(\mathbf{a}) \triangleq \left\{ j \in [f] \mid \mathbf{a} \cap \mathcal{F}_j \neq \emptyset \right\}
\end{equation*}
and the size of the support family as
\begin{equation*}
    b(\mathbf{a}) \triangleq |\mathcal{B}(\mathbf{a})| \in \big[\lceil{\frac{d}{s}}\rceil,d\big].
\end{equation*}
\end{definition}
We also let \( \mathbf{A}_{\mathrm{full}} \triangleq \{ \mathbf{a}\in \mathbf{A}_{n,d} \ : \ b(\mathbf{a}) = d\}\) represent the so-called set of full support (maximal support) $d$-tuples, and let \(\mathbf{A}_{\mathrm{com}}=\mathbf{A}_{n,d}\backslash \mathbf{A}_{\mathrm{full}}\) be its complement set.

Now, we must first partition \(\mathbf{A}_{\mathrm{full}}\) into $N'=\binom{f}{d}$ groups \textcolor{black}{\(\widetilde{\mathbf{\Phi}}^{\mathrm{(full)}}_\sigma\)}, \(\sigma\in \binom{[f]}{d}\). We design these groups as follows
\begin{equation}
\label{eqfull}
    \widetilde{\mathbf{\Phi}}^{\mathrm{(full)}}_\sigma=
    \{\mathbf{a}\in \mathbf{A}_{n,d}\ |\ \mathcal{B}(\mathbf{a})=\sigma\}
\end{equation}
where it is easy to verify that 
\begin{equation}
    \mathbf{A}_{\mathrm{full}} =\bigcup \limits_{\sigma\in \binom{[f]}{d}}\widetilde{\mathbf{\Phi}}^{\mathrm{(full)}}_\sigma
\end{equation}
i.e., that the groups $\widetilde{\mathbf{\Phi}}^{\mathrm{(full)}}_\sigma$ are disjoint and cover $\mathbf{A}_{\mathrm{full}}$.  
Finally, we note that \(|\widetilde{\mathbf{\Phi}}_\sigma^{\mathrm{(full)}}|=s^d\) for every \(\sigma\in \binom{[f]}{d}\), which implies that \(|\mathbf{A}_{\mathrm{full}}|= s^d\cdot \binom{f}{d}\), and thus that \(|\mathbf{A}_{\mathrm{com}}|=\binom{n}{d}-s^d\cdot\binom{f}{d}\).  

Now, it remains to partition \(\mathbf{A}_{\mathrm{com}}\) into \(N'\) groups. In this regard, we consider a partition of \(\mathbf{A}_{n,d}\), which classifies its \(d\)-tuples according to the size of their support family. This partition is as follows
\begin{equation}
    \mathbf{A}_{n,d}=\{\mathbf{C}_{\lceil{\frac{d}{s}}\rceil},\mathbf{C}_{\lceil{\frac{d}{s}}\rceil+1},\dots,\mathbf{C}_d\}
\end{equation}
where for each \(\beta\in [\lceil{\frac{d}{s}}\rceil,d]\), the set
\begin{equation}
\label{eq:C_beta}
    \mathbf{C}_\beta \triangleq \{\mathbf{a}\in \mathbf{A}_{n,d}\mid b(\mathbf{a})=\beta\}
\end{equation}
represents the set of \(d\)-tuples \( \mathbf{a}\in\mathbf{C}_\beta \) that each intersects exactly \( \beta \) families. Naturally, we have \(\mathbf{C}_d= \mathbf{A}_{\mathrm{full}}\) and \(\mathbf{A}_{\mathrm{com}}=\bigcup_{\beta=\lceil\frac{d}{s}\rceil}^{d-1}\mathbf{C}_\beta\). Let us fix 

 \(\beta \in [\lceil{\frac{d}{s}}\rceil,d-1]\). Then, for each \( \mathcal{I} \in \binom{[f]}{\beta} \), we define
    \begin{equation}
    \label{eq:C_beta, I}
          \mathbf{C}_{\beta,\mathcal{I}} \triangleq \left\{ \mathbf{a} \in \mathbf{C}_\beta \;\middle|\; \mathcal{B}(\mathbf{a})=\mathcal{I}\right\}
    \end{equation}
where we will see that \( |\mathbf{C}_{\beta,\mathcal{I}} | \) remains fixed for all \( \mathcal{I} \in {[f] \choose \beta} \) (cf. Appendix~\ref{proofoflem:count}).

 We define \(\mathbf{C}_{\beta, \mathcal{I}, \sigma }\) as a subset of \(\mathbf{C}_{\beta, \mathcal{I}}\) that are allocated to group \(\widetilde{\mathbf{\Phi}}_\sigma\). In Appendix \ref{appendix: C_beta,I}, we proceed with a sequence of steps that will lead to the creation of the \(\mathbf{C}_{\beta, \mathcal{I}, \sigma}\). Regarding the $\mathbf{C}_{\beta,\mathcal I,\sigma}$, we recall that for each \(\sigma \in {[f] \choose d}\), there exist \({d \choose \beta}\) distinct \(\mathcal{I} \in {[f] \choose \beta}\) such that \(\mathcal{I} \subset \sigma\). Consequently, 
   \begin{equation}
       \label{eq: phi_com}
       \widetilde{\mathbf{\Phi}}_{\sigma}^{\mathrm{(com)}}=\bigcup_{\beta=\lceil \frac{d}{s}\rceil}^{d-1}\bigcup_{\mathcal{I}\subset \sigma}\mathbf{C}_{\beta, \mathcal{I}, \sigma}.
   \end{equation}
At this point, we have a complete description of the subfunctions (\(d\)-tuples) allotted --- for now --- to worker \(\sigma\), and these subfunctions take the form
 \begin{equation}
     \label{eq:SubfunctionCase1}
     \widetilde{\mathbf{\Phi}}_{\sigma}=\widetilde{\mathbf{\Phi}}_{\sigma}^{(\mathrm{full})}\cup \widetilde{\mathbf{\Phi}}_{\sigma}^{\mathrm{(com)}}.
 \end{equation}
We now continue with Lemma \ref{lem:ub-lb}, which bounds the number of subfunctions associated --- initially --- to each of the first \(N' = \binom{k}{d}\) servers.
     \begin{lemma}
         \label{lem:ub-lb}
         For Case 1 (Section~\ref{subsec: Case1}), the construction entails a $|\widetilde{\mathbf{\Phi}}_{\sigma}|, \ \sigma \in \binom{[f]}{d}, $ that is bounded as follows 
         \[
         \frac{{n \choose d}}{N'}-2^d+d\le|\widetilde{\mathbf{\Phi}}_{\sigma}|\le \frac{{n \choose d}}{N'}+2^d-d.
         \]
     \end{lemma}
     \begin{proof}
         The proof of Lemma~\ref{lem:ub-lb} is provided in Appendix~\ref{proofoflem:ub}.
     \end{proof}

The following lemma describes the communication cost for the same design. 

\begin{lemma}
\label{lem:lemma_pi_g=0}
For Case 1 (Section~\ref{subsec: Case1}), for any chosen integer \(s\) such that $s \mid n $, and for $k = n/s$, the IC design achieves
\[
\pi = s\cdot d.
\]

\end{lemma}

\begin{proof}
The proof of Lemma \ref{lem:lemma_pi_g=0} is provided in Appendix \ref{Proofoflem: pi_g=0}.
\end{proof}
We now proceed with the second case.

\subsection{Case 2: $k \nmid n $}
\label{subsec: Case2}
% N_max for Setting 2, N<N_max , sparse, d<=n naturally
The scheme here will follow the general format of the scheme in Case 1 (Section~\ref{subsec: Case1}), except that we will modify some of our parameters, to adapt to the fact that $k \nmid n $. In particular, for $k,f$ given, as before, by
\begin{equation}
\label{eq:largest_r2}
    k = \max \left\{ r \in \mathbb{Z}^{+} \,\mid\, \binom{r}{d} \leq N \right\}, \ \ \ f=k
\end{equation}
and for \begin{equation}
\label{eq:s_0}
    s_0 = \left\lfloor \frac{n}{k + d} \right\rfloor + 1
\end{equation}
we will temporarily discard 
\begin{equation}
\label{eq:g,f}
    g= n - k \cdot s_0
\end{equation}
files, and will proceed under a consideration that the new number of files is \( n'= n -g \). With this new number of files $n'$ in place, we will again initially consider that we have
\(N' = \binom{k}{d}\) workers. In this second case, as discussed in Appendix~\ref{Appendix: limit on N}, we will accept the limitation that 
\textcolor{black}{\( N\le (\frac{9}{10}\sqrt{\frac{n}{d}})^{d}\)}.

For the above parameters \( n' = n-g,f, d \), and \( N' = \binom{f}{d} \), the first step is to employ the design of Section~\ref{subsec: Case1}, to partition \( \mathbf{A}_{n',d} \) into \( N' \) groups \(\widetilde{\mathbf{\Phi}}_{\sigma}^{\mathrm{(full)}}\cup \widetilde{\mathbf{\Phi}}_{\sigma}^{\mathrm{(com)}}, \ \sigma \in \binom{[f]}{d}\). The second step considers the excluded \footnote{For example, for $n=5, n'=4,d=2$, we have that $\mathbf{A}_{\mathrm{exc}} = \{15,25,35,45\}$.}\(d\)-tuples 
\begin{equation}
\label{eq:A_dis}
\mathbf{A}_{\mathrm{exc}} = \mathbf{A}_{n,d} \setminus \mathbf{A}_{n',d}
\end{equation}
as well as the set of \( g \) \emph{excluded elements} \( [n] \setminus [n'] \), which we here denote as
\begin{equation}
\label{eq:D}
    \mathcal{E}\triangleq\{n'+1, \dots, n'+g = n\}.
\end{equation}
Any \(d\)-tuple \(\mathbf{t} \in \mathbf{A}_{\mathrm{exc}}\) will have an arbitrary number \(m_{\mathbf{t}}=|\mathbf{t}  \cap \mathcal{E}|\) of components/elements from the excluded file-index set $\mathcal{E}$, and it will have  \(d-m_{\mathbf{t}} = |\mathbf{t}  \cap [n']|\) elements from the rest. 
It is easy to see that  \( m\in [1, \min\{d,g\}]\) and thus that \(d-m \in [\max\{d-g,0\},d-1]\).  Whenever there is no ambiguity, we will henceforth revert to the simpler notation \(m\) instead of $ m_{\mathbf{t}}$.

For every \(m\in [1, \min({d,g})]\), we define the set
 \begin{equation} 
    \label{eq:R_m,beta}
    \mathbf{R}_{m,\beta} \triangleq \left\{ \mathbf{t} \in \mathbf{A}_{\mathrm{exc}} \;\middle|\; b(\mathbf{t}) = \beta  ,\ |\mathbf{t}  \cap \mathcal{E}|=m \right\}
    \end{equation}
which describes the \(d\)-tuples \( \mathbf{t} \) that intersect exactly \( \beta \) families and contain \(m\) excluded elements from $\mathcal{E}$. Notice that \(\beta\) can take values in the range \(\big[\lceil\frac{d-m}{s_0}\rceil, d-m\big]\). If \(m=d\leq g\), then \(\beta =0\), which means that all the entries of \(\mathbf{t}\) are from \(\mathcal{E}\). 

Let us now partition \(\mathbf{A}_{\mathrm{exc}}\)  as follows
    \[
    \mathbf{A}_{\mathrm{exc}}=\bigcup_{m=1}^{\min\{d, g\}}\bigcup_{\beta=\lceil\frac{d-m}{s_0}\rceil}^{d-m} \mathbf{R}_{m,\beta}.
    \]

    For each $\mathcal{I}\in\binom{[f]}{\beta}$, let us now define 
 \begin{equation}
  \label{eq:R_beta,I}   
 \mathbf{R}_{\beta,\mathcal{I}}
 \triangleq \left\{\mathbf{t}\in\mathbf{A}_{\mathrm{exc}}\ \mid\ \mathcal{B}(\mathbf{t})=\mathcal{I},\ 
 \mathbf{t}\in\bigcup_{m=1}^{\min\{d-\beta,g\}}\mathbf{R}_{m,\beta}\right\}
  \end{equation}
 to be the set of all d-tuples $ \mathbf{t} \in \mathbf{A}_{\mathrm{exc}}$ that intersect exactly all families in $\mathcal{I}$, where in the above, $\mathcal{B}(\mathbf{t})$ denotes the set of families that $\mathbf{t}$ intersects. It is relatively easy to see that the cardinality \(|\mathbf{R}_{\beta,\mathcal{I}}|\) remains the same for any $\mathcal{I}\in\binom{[f]}{\beta}$ (see Appendix \ref{appendix: size of t_beta}).
 Let us now also define 
 \begin{equation}
 \label{eq:R_beta}
  \mathbf{R}_{\beta}\triangleq\bigcup_{\mathcal{I}\in\binom{[f]}{\beta}}\mathbf{R}_{\beta,\mathcal{I}} \subset \mathbf{A}_{\mathrm{exc}}
  \end{equation}
to be the set of all excluded \(d\)-tuples that meet exactly $\beta$ families.
Furthermore, directly by applying the established ranges of parameters \(m\) and \(k\), we can conclude that the range of \(\beta \in [\beta_{\mathrm{min}}, \beta_{\mathrm{max}}]\), is defined by
\begin{equation}
\label{eq:beta_min}
    \beta_{\min} \;\triangleq\; \left\lceil \frac{d-\min\{d,g\}}{s_0} \right\rceil
            \;=\; \left\lceil \frac{\max\{0,d-g\}}{s_0} \right\rceil,
\end{equation}
\begin{equation}
\label{eq:beta_max}
    \beta_{\max} \;\triangleq\; d-1.
\end{equation}
   Our next step involves going through the range of $\beta$. For each \(\beta \in [\beta_{\mathrm{min}}, \beta_{\mathrm{max}}]\), we partition each time the set \(\mathbf{R}_{\beta}\) into \( N'={f \choose d} \) groups.  This partitioning is described in detail in Appendix \ref{appendix: R_beta,I}. In particular, let us first recall that each group is labeled by a \(\sigma \in {[f] \choose d}\). For each such $\sigma$, there exist \({d \choose \beta}\) different subsets \(\mathcal{I} \subset\sigma\) with cardinality \(\beta\). For each \(\mathcal{I}\subset \sigma\), the set \(\mathbf{R}_{\beta, \mathcal{I}, \sigma}\) collects all \(d\)-tuples in \(\mathbf{R}_{\beta, \mathcal{I}}\) associated to group $\sigma$ --- again as described in Appendix \ref{appendix: R_beta,I}. We then form the union
   \begin{equation}
       \label{eq: phi_exc}
       \widetilde{\mathbf{\Phi}}_{\sigma}^{\mathrm{(exc)}}=\bigcup_{\beta=\beta_{\mathrm{min}}}^{\beta_{\mathrm{max}}}\bigcup_{\mathcal{I}\subset \sigma}\mathbf{R}_{\beta, \mathcal{I}, \sigma}
   \end{equation}
and then the union 
\[\mathbf{A}_{\mathrm{exc}}=\bigcup_{\sigma}\widetilde{\mathbf{\Phi}}_{\sigma}^{\mathrm{(exc)}}. \] Recall now that  \(\mathbf{A}_{n',d}\) is already partitioned as 
\[\mathbf{A}_{n',d} = \bigcup \limits_{\sigma\in \binom{[f]}{d}} \bigl(\widetilde{\mathbf{\Phi}}_{\sigma}^{\mathrm{(full)}}\cup\widetilde{\mathbf{\Phi}}_{\sigma}^{\mathrm{(com)}}\bigr) \]
and that each \(\widetilde{\mathbf{\Phi}}_{\sigma}^{\mathrm{(full)}}\cup\widetilde{\mathbf{\Phi}}_{\sigma}^{\mathrm{(com)}} \subset \mathbf{A}_{n',d}\) is placed in group $\sigma \in {[f] \choose d}$. Consequently, in the end, we have
\begin{equation} \label{eq:Partition1a}
\mathbf{A}_{n,d}=\mathbf{A}_{n',d}\cup \mathbf{A}_{\mathrm{exc}}=\bigcup_{\sigma}(\widetilde{\mathbf{\Phi}}_{\sigma}^{\mathrm{(full)}}\cup \widetilde{\mathbf{\Phi}}_{\sigma}^{\mathrm{(com)}}\cup\widetilde{\mathbf{\Phi}}_{\sigma}^{\mathrm{(exc)}})=\bigcup_{\sigma}\widetilde{\mathbf{\Phi}}_{\sigma}
\end{equation}
where again, $\widetilde{\mathbf{\Phi}}_{\sigma}$ is the entire set of $d$-tuples assigned, as a first step, to group $\sigma$, up to the $N'$th group.

The following lemma describes the communication cost associated to the above task allocation. We also briefly recall that the design in this second case applies for $N\leq (\frac{9}{10}\sqrt{\frac{n}{d}})^{d}$, 
as well as recall that $ s_0 = \left\lfloor \frac{n}{k + d} \right\rfloor + 1$ and $g= n - k \cdot s_0$.

\begin{lemma}
\label{lem:lemma_pi_g>0}
For \( n \), \( d\), the IC design in Case 2 (Section~\ref{subsec: Case2}) achieves
\[
\pi = s_0\cdot d +g.
\]
\end{lemma}
\begin{proof}
The proof of Lemma \ref{lem:lemma_pi_g>0} is provided in Appendix \ref{Proofoflem: pi_g>0}.
\end{proof}
We continue with bounding the number of elements per group. We briefly recall that \( N'={f \choose d} \). % Lemma \ref{lem:ub-lb2}.
% q<=d
     \begin{lemma}
         \label{lem:ub-lb2}
         In the IC design in Case 2 (Section~\ref{subsec: Case2}), the following bound holds
         \[
         \frac{{n \choose d}}{N'}-2^{d+1}+2d\le|\widetilde{\mathbf{\Phi}}_\sigma|\leq \frac{{n \choose d}}{N'}+2^{d+1}-2d.
         \]
     \end{lemma}
     \begin{proof}
         The proof is provided in Appendix~\ref{proofoflem:ub2}.
         \end{proof}
         
\subsection{Extension of the Partition from \(N'\) Groups to \(N\) Groups}
\label{subsec: N' to N}
Recall (cf.~\eqref{eq:Partition1a}) that we have already partitioned \(\mathbf{A}_{n,d}\) into \(N' = \binom{f}{d}\) disjoint groups
\[
\widetilde{\mathbf{\Phi}}_{\sigma_1},\dots,\widetilde{\mathbf{\Phi}}_{\sigma_{N'}},
\qquad \sigma_1,\dots,\sigma_{N'}\in\binom{[f]}{d}.
\]
We will here redistribute the \(d\)-tuples of these $N'$ groups across all existing $N$ groups. 

Towards this, let us assume that the indices \(\sigma_1,\dots,\sigma_{N'} \) are in lexicographic order and, in order to ease notation, let us rename the corresponding \(N'\) groups by their lexicographic position, as follows
\begin{equation}
  \widetilde{\mathbb{P}}\triangleq \{\widetilde{\mathbf{\Phi}}_1,\dots,\widetilde{\mathbf{\Phi}}_{N'} \}
\end{equation}
where in particular, \(\widetilde{\mathbf{\Phi}}_b=\widetilde{\mathbf{\Phi}}_{\sigma_b}\) for \(b \in[N']\). 
Recalling again that there are \(N\ge N'\) actual groups, let us first define the following variables
\begin{equation}
\label{eq: q,p,r}
  q\triangleq\Big\lfloor\frac{N}{N'}\Big\rfloor,\qquad
p\triangleq\Big\lceil\frac{N}{N'}\Big\rceil,\qquad
r\triangleq N\bmod N' \quad(0\le r<N')
\end{equation}
thus noting that \(
N=qN'+r,\) where \(p=q\ \text{if }r=0\), and  \(p=q+1\ \text{if }r>0.\) 

At this point, we proceed with the first step of dividing the \(d\)-tuple set of each of the first $N'$ groups into different parts, and then with the second step of redistributing some of these parts to fill up the empty $N-N'$ groups.
\paragraph*{Step 1 -- Dividing the \(d\)-tuples of each of the first $N'$ groups}   
For each \(b\in[N']\), we define the number of parts 
\[
s_b\triangleq\begin{cases}
p & \text{if } 1\le b\le r,\\[4pt]
q & \text{if } r<b\le N'
\end{cases}
\]
and we split each \(\widetilde{\mathbf{\Phi}}_b\) into \(s_b\) disjoint sub-parts using lexicographic ordering that yields slicing of equal sizes, plus\footnote{We keep track of the exact size of each sub-part.} or minus $1$. 
We denote these sub-parts by 
\[
\widetilde{\mathbf{\Phi}}_b^{(0)},\widetilde{\mathbf{\Phi}}_b^{(1)},\dots,\widetilde{\mathbf{\Phi}}_b^{(s_b-1)},
\qquad \widetilde{\mathbf{\Phi}}_b=\bigcup_{j=0}^{s_b-1}\widetilde{\mathbf{\Phi}}_b^{(b')}.
\]
\paragraph*{Step 2 -- Extending to \(N\) groups} 
We then relabel these sub-parts to obtain the desired \(N\) groups.  
We define the new \(N\) groups \(\widetilde{\mathbf{\Phi}}_1,\dots,\widetilde{\mathbf{\Phi}}_N\) by the indexing rule
\[
\widetilde{\mathbf{\Phi}}_{\,b + b'N'} \;\triangleq\; \widetilde{\mathbf{\Phi}}_b^{(b')},
\qquad\text{for } b\in\{1,\dots,N'\},\; b'\in\{0,\dots,s_b-1\}\] where we can very easily check the following.
\begin{enumerate}
   \item \emph{Indexing bijection:} The mapping \((b,b')\mapsto t=b+b'N'\) is a bijection between
  \(\{1,\dots,N'\}\times\{0,\dots,s_b-1\}\) and \(\{1,\dots,N\}\), so every new index \(t\in\{1,\dots,N\}\) corresponds to exactly one subpart. To see this, we simply note that index \(b+b'N'\) `runs across' all the integers \(1,\dots,N\) because when \(1\le b\le r\) we have \(b'\in\{0,\dots,p-1\}\) and the largest index is \(r+(p-1)N'=qN'+r=N\); when \(r<b\le N'\) we have \(b'\in\{0,\dots,q-1\}\) and the largest index is \(N' + (q-1)N' = qN' \le N\).
  \item \emph{Partition:} By construction, \(\widetilde{\mathbf{\Phi}}_1,\dots,\widetilde{\mathbf{\Phi}}_N\) are pairwise disjoint and cover 
\(\mathbf{A}_{n,d}\) thus guaranteeing the partition
  \begin{equation} \label{eq:PartitionBeforeLast}
  \mathbf{A}_{n,d}=\bigcup_{b=1}^N \widetilde{\mathbf{\Phi}}_b.
  \end{equation}
  This directly holds because previously the group set $\widetilde{\mathbf{\Phi}}_{\sigma_1},\dots,\widetilde{\mathbf{\Phi}}_{\sigma_{N'}},$ partitioned \(\mathbf{A}_{n,d}\), and because the redistribution (i.e., the transition from $N'$ to $N$ groups) guarantees, by design, that no entries from these first $N'$ groups appears twice.
 
\end{enumerate}
At this point, we have the following two lemmas; the main lemma being Lemma~\ref{lem: delta pi, A,n,d}, which needs the following lemma that bounds the variables $q,p,r$ from \eqref{eq: q,p,r}, as they pertain to the IC design. 
\begin{lemma} 
\label{lem:N,N'}
  For any \(n,d,N\), the IC design's use of $N' = \binom{k}{d}$ and $k = \max \left\{ r \in \mathbb{Z}^{+} \,\middle|\, \binom{r}{d} \leq N \right\}$ (cf.~ \eqref{eq:largest_r}), guarantees that 
  \[
   \frac{N}{N'}< d+1\le 2^d.
    \]
\end{lemma}
\begin{proof}
The proof of Lemma \ref{lem:N,N'} is provided in Appendix \ref{proofoflem: N, N'}.
\end{proof}
From Case 1 and Case 2, we have the following upper bound guarantees on the \(\delta\) and \(\pi\) achieved by the IC design.
\begin{lemma}
\label{lem: delta pi, A,n,d}
Given \(n\), \(d \le \frac{n}{32}\), and \(N \le (\frac{9}{10}\sqrt{\frac{n}{d}})^{d} \), the partition $\mathbf{A}_{n,d}=\bigcup_{b=1}^N \widetilde{\mathbf{\Phi}}_b$ from \eqref{eq:PartitionBeforeLast}, guarantees 
\[
\delta=\frac{\max_{b\in[N]}|\widetilde{\mathbf{\Phi}}_b|}{\lceil{n \choose d}/N \rceil}\le 4
\]
and 
\[
\pi\le \frac{4e \cdot n}{N^{\frac{1}{d}}}.
\]
\end{lemma}
\begin{proof}
    The proof is found in Appendix~\ref{Proofoflem: delta pi, A,n,d}.
\end{proof}

\subsection{Partition Refinement }
\label{sec: Random task}
In the above, we have provided a partition \(\widetilde{\mathbb{P}} = \{ \widetilde{\mathbf{\Phi}}_1, \widetilde{\mathbf{\Phi}}_2, \ldots, \widetilde{\mathbf{\Phi}}_N \}\) 
of \(\mathbf{A}_{n,d}\), and offered guarantees on the achievable $\pi$ and $\delta$.  Now we need to consider the actual subfunction set \(\mathbf{X} \subseteq \mathbf{A}_{n,d}\), which means that we need to partition $\mathbf{X}$, and this partition will yield the file allocation (corresponding to a cost $\pi_{\mathbf{X}}$), as well as the subfunction allocation (corresponding to a cost $\delta_{\mathbf{X}}$). The file allocation is automatic; we simply maintain the exact file allocation suggested by the above partition of \(\mathbf{A}_{n,d}\) and thus maintain the exact same $\pi$. The subfunction allocation across the $N$ servers will be defined by the new partition, now of $\mathbf{X}$, and this is simply the intersection of the above partition $\widetilde{\mathbb{P}}$ with $\mathbf{X}$. While this step is simple enough, it runs the risk of yielding a very large $\delta_\mathbf{X}$, depending on $\mathbf{X}$. Thus, our main effort now will be to bound $\delta_{\mathbf{X}}$.

To analyze this, we can view \(\mathbf{X}\) as a realization obtained from a random thinning of \(\mathbf{A}_{n,d}\), where each \(d\)-tuple \(\mathbf{a} \in \mathbf{A}_{n,d}\) 
is independently included in \(\mathbf{X}\) with probability \(\varphi=|\mathbf{X}| / \binom{n}{d}\). Let us proceed to bound $\delta_{\mathbf{X}}$.

Recalling the original partition 
\(\widetilde{\mathbb{P}} = \{ \widetilde{\mathbf{\Phi}}_1, \widetilde{\mathbf{\Phi}}_2, \ldots, \widetilde{\mathbf{\Phi}}_N \}\) 
of \(\mathbf{A}_{n,d}\), we will now assign to each server $b\in [N]$, the entries of the group
\begin{equation}
    \label{eq: sampled group size}
    \mathbf{\Phi}_{b}=\widetilde{\mathbf{\Phi}}_b \cap \mathbf{X}
\end{equation}
based on which we define 
\begin{equation}
\label{eq: random, mu_max}
\mu_b\triangleq| \mathbf{\Phi}_{b}|, \ b\in[N], \ \ \mu_{\min} \triangleq \min_{b\in[N]}\mu_b, \quad \mu_{\max} \triangleq \max_{b\in[N]}\mu_b
\end{equation}
thus rewriting 
\begin{equation}
\label{eq: random, delta_x}
 \delta_{\mathbf{X}}  =  \frac{\max_{b} \mu_b}{\lceil|\mathbf{X}|/N\rceil}=\frac{\mu_{\mathrm{max}}}{\lceil(\varphi\cdot {n \choose d})/N\rceil} .  
\end{equation}

\textcolor{black}{We proceed with the following lemma, which pertains to the IC design partition of \(\mathbf{X}\), and which holds for all \(\varphi\ge \varphi_{\mathrm{min}} \approx \frac{ \ln n}{n^{d/2}}\) (see later cf.\eqref{eq:varphi_min}). Note also that $\varphi_{\mathrm{min}}$ is also upper bounded by 
$C_d\cdot\frac{ \ln n}{n^{d/2}}$, where \(C_d=192 \cdot (e\cdot d)^{d/2}\).} 

\begin{lemma}
\label{lem: conc on group size}
For \(n\ge 32d\), \(N\le (\frac{9}{10}\sqrt{\frac{n}{d}})^{d} \), and 
 $\varphi\ge \varphi_{\mathrm{min}} ,$ then with probability at least \(1-\tfrac1n\), we have 
\[
\max_{b \in [N]}|\mathbf{\Phi}_b| \;\le\; \frac{5}{4}\varphi \max_{b \in [N]}|\widetilde{\mathbf{\Phi}}_b|
\]and thus we have 
\[
\delta_{\mathbf{X}}\le 5.
\]
\end{lemma}
 This concludes the proof of the main theorem, showing that the IC design, both for Case 1 and Case 2, guarantees \(\delta_{\mathbf{X}}\le 5\) with probability \(1-\tfrac1n\), while also recalling that it guarantees \(\pi_{\mathbf{X}}\le  \frac{4e\cdot n}{N^{1/d}}\).

\section{Comparison}
\label{sec:comparison}
In this section, we compare the performance of the proposed IC design, with existing graph partitioning algorithms in terms of their average replication factor guarantees. 
Although \(\mathrm{ARF}\) minimization is not the primary objective of the IC design, we have seen that 
\begin{equation}\label{eq:piToARF}
    \mathrm{ARF} \leq \frac{N\pi}{n}
\end{equation}
    and thus what we minimize here is effectively an upper bound on the \(\mathrm{ARF}\). 
To the best of our knowledge, theoretical guarantees on achievable \(\mathrm{ARF}\) are limited to the case of $d=2$. We restrict our comparisons to this case, and focus on the state-of-art results reported in~\cite{Dynamic,TrillionEdges,Zhang2017GraphEdgePartitioning,ProjectivePlane}. 

For the IC design, for Case~1 (where the divisibility conditions of Lemma~\ref{lem:lemma_pi_g=0} hold), we have 
$\pi = sd = 2n/k$ (recall Lemma~\ref{lem:lemma_pi_g=0} and recall that $s=n/k$),  where we also recall that $k$ is such that $N=\binom{k}{2}$, which means that
\(
k = \frac{1+\sqrt{\,8N+1\,}}{2}, 
\)
which means that $\pi = \frac{4n}{1+\sqrt{8N+1}}<\frac{\sqrt{2}n}{\sqrt{N}}
,$ which finally says (cf.~\eqref{eq:piToARF}) that
\[
\mathrm{ARF}^{(\mathrm{IC})} < \sqrt{2N}.
\]
Similarly, for Case 2, %In the general case, where the above divisibility conditions are not necessarily met, 
recall that Lemma~\ref{lem:lemma_pi_g>0} gives $\pi = s_{0}d + g$, where $s_{0}=\lfloor n/(k+d)\rfloor+1$ and $g=n-s_{0}k$. Combined with  
\(
\frac{k(k-1)}{2} \le N \le \frac{k(k+1)}{2}
\)
i.e., combined with  
\[
\frac{-1+\sqrt{\,8N+1\,}}{2} \;\le\; k \;\le\; \frac{1+\sqrt{\,8N+1\,}}{2}
\]
and after setting $k=\frac{-1+\sqrt{\,8N+1\,}}{2}$, we obtain
\[
\pi \le \frac{4n}{k+2}+2-k
\le \frac{2\sqrt{2}\,n}{\sqrt{N}} + \frac{5-\sqrt{8N+1}}{2}
\]
and since  \(\frac{5-\sqrt{8N+1}}{2}\leq 0\) for all $N \geq 3$, we get $\pi \le \frac{2\sqrt{2}\,n}{\sqrt{N}}$
which, for every \(N \geq 3\), gives 
\begin{equation} \label{eq:comparisonCase2}
\mathrm{ARF}^{(\mathrm{IC})} \le 2\sqrt{2N}.
\end{equation}

Let us now compare $\mathrm{ARF}^{(\mathrm{IC})}$ with the guarantee on the achievable \(\mathrm{ARF}^{(\mathrm{Dynamic})}\) of the algorithms in \cite{Dynamic,TrillionEdges}, from where we know that 
\[
\mathrm{ARF}^{(\mathrm{Dynamic})} \le \frac{n + |\mathbf{X}| + N}{n}.
\]
Equating the general-case bound of $\mathrm{ARF}^{(\mathrm{IC})}$ in \eqref{eq:comparisonCase2} with $\mathrm{ARF}^{(\mathrm{Dynamic})}$, and solving for
$\varphi = |\mathbf{X}|/\binom{n}{2}$, yields
\begin{equation} \label{eq:comparisonDynamic1}
    \mathrm{ARF}^{(\mathrm{IC})} \le \mathrm{ARF}^{(\mathrm{Dynamic})}
\quad\Longleftrightarrow\quad
\varphi \ge \frac{2}{n-1}\!\left(2\sqrt{2N} - 1 - \frac{N}{n}\right)
\end{equation}
where we emphasize one more time that this above condition on \(\varphi\) is a result of comparing existing \emph{guarantees} on achievable performance; the above precisely characterizes when our IC design gives a better guarantee on the achievable \(\mathrm{ARF}\) than the guarantees of the schemes in \cite{Dynamic,TrillionEdges}. Let us recall that for the case of $d=2$ here, the smallest $\varphi$ scales with $n^2$. Looking at~\eqref{eq:comparisonDynamic1}, we can readily conclude that when \(N\) is small compared to \(n\) (in which case, the expression \(2\sqrt{2N} - 1 - \tfrac{N}{n}\) becomes less significant), the IC design performs better for all $\varphi$ that scale bigger than $1/n$, whereas, on the other extreme of having $N\approx n$ (this is an extreme because, for $d=2$, we ask for \(N\leq (\frac{9}{10}\sqrt{\frac{n}{2}})^2<n\)), we can see that the IC design performs better for all $\varphi$ that scale bigger than \(1/\sqrt{n}\).

Similar performance is also recorded in \cite{Zhang2017GraphEdgePartitioning}, which presents an algorithm that is based on a neighborhood heuristic, and which has an achievable \(\mathrm{ARF}\) that satisfies
\[
\mathrm{ARF}^{(\mathrm{NH})} \le 2\sqrt{N} + \frac{N}{n}
\]
which comes close to the IC design's guarantees of $\mathrm{ARF}^{(\mathrm{IC})} < \sqrt{2N}$ and $\mathrm{ARF}^{(\mathrm{IC})} < 2\sqrt{2N}$, respectively for Case~1 and Case~2.

Finally, it is interesting to compare to the projective-plane-based construction of \cite{ProjectivePlane}, which is though applicable only when $N = q^{2}+q+1$ for a prime power $q$, and for which the guaranteed 
\[
1.5\sqrt{N} \le \mathrm{ARF}^{(\mathrm{Projective})} \le 2\sqrt{N}
\]
again comes close to our own guarantees.

To conclude, we emphasize that our primary design objective is to minimize \(\pi\); \(\mathrm{ARF}\) is a derived metric, yet even under this secondary measure, the IC design offers guarantees that are competitive with --- and often superior to --- the guarantees for the existing edge partitioning schemes, for a broad range of parameters.  

\section{Conclusions}
\label{sec: conclu}
This paper introduced the Interweaved Clique design, a deterministic construction for data and task allocation in a general distributed computing setting. In a setting where the desired function is modeled as admitting a decomposition into subfunctions indexed by a set $\mathbf{X} \subseteq \mathbf{A}_{n,d}$ --- where each element of $\mathbf{X}$ corresponds to a $d$-tuple of the input files of a subfunction --- the IC design provides an explicit partition of $\mathbf{X}$ across $N$ workers and yields a communication cost $\pi_{\mathbf{X}}$ whose scaling is optimal for all $\mathbf{X}$ having non-vanishing normalized size $\varphi$. The same construction also controls the computation-delay parameter $\delta_{\mathbf{X}}$, and for all task sets with non-vanishing normalized size, this delay remains bounded by a constant with high probability. 

In addition to the IC design, we have a converse which, albeit simple, it is also tight in the sense that it is the tightest (largest) lower bound that holds true for all $\mathbf{X}$. In particular, for given \(\varphi \in (0,1]\), the deterministic lower bound \(\pi_{lb} \triangleq \varphi^{1/d}\frac{n}{N^{1/d}}\) is tight in a sense that there exists a \(\mathbf{X} \subseteq{[n] \choose d}\) that achieve \(\pi_{\mathbf{X}} \asymp \pi_{lb} \). Intuitively, such \(\mathbf{X}\) is a union of \(N\) disjoint cliques where each clique \(\mathcal{C}_i\), \(i \in [N]\) formed by at most \(\pi_{lb}\) elements such that each pair of these cliques can have \(d-1\) elements in common (partition condition) , i.e.,  \(|\mathcal{C}_i\cap\mathcal{C}_j|\le d-1\). This confirms that when \(\mathbf{X}\) is structured as a collection of dense, disjoint cliques (\(|\mathcal{C}_i\cap\mathcal{C}_j|\le d-1\)), the maximum number of digits per group meets the theoretical minimum. 

Together, these results characterize the \emph{order-optimal} communication and computation performance, $\pi_{\mathbf{X}}$ and $\delta_{\mathbf{X}}$, achievable under file and task assignments with a fixed subfunction degree~$d$. It is worth commenting here that this choice of $\pi_{\mathbf{X}}$ as the communication metric differs from the average replication factor commonly used in the hypergraph-partitioning literature. Recall that \(\mathrm{ARF}\) measures the average number of workers to which each file is replicated, whereas $\pi_{\mathbf{X}}$ measures the maximum number of files transmitted to any worker and thus captures the worst-case load on the master--worker links in systems with parallel, equal-capacity communication channels. This distinction becomes critical when computational-delay $\delta_{\mathbf{X}}$ constraints are absent: \(\mathrm{ARF}\) becomes degenerate, as one may assign all files and all subfunctions to a single worker, yielding $\mathrm{ARF}=1$ while incurring $\pi_{\mathbf{X}} = n$, the largest possible communication cost. Hence, \(\mathrm{ARF}\) does not reflect the operational communication bottleneck in such settings. In contrast, $\pi_{\mathbf{X}}$ remains meaningful irrespective of any computation-delay considerations and directly characterizes the worst-case communication load.

Let us also recall certain structural properties of the IC design which facilitate practical deployment. Firstly, the construction is deterministic and does not rely on search-based procedures, keeping implementation complexity minimal. Secondly, and very importantly, the designed file allocation remains independent of the particular task decomposition $\mathbf{X}$. Once files are placed across the $N$ workers, the same placement can be used for any $\mathbf{X} \subseteq \mathbf{A}_{n,d}$; only the subfunction assignment changes from one task to another. As a result, introducing a new task set does not require relocating files across workers, and the performance guarantees remain valid across these sets.   

In the end, the near-unified nature of the design, and its simplicity, allows for the first time, for the derivation of fundamental limits that hold for pertinent parameter regimes, and which yield tight scaling laws for a broad range of scenarios of non-vanishingly small $\mathbf{X}$.

\begin{appendices}
\label{sec: appendixs}

\section{Proof of Lemmas}
\label{proofoflemmas}

\subsection{Proof of Lemma \ref{lem:ub-lb}}
\label{proofoflem:ub}
From \eqref{eq:SubfunctionCase1} in Section \ref{subsec: Case1}, we recall that for each \(\sigma \in {[f] \choose d}\), the set $\widetilde{\mathbf{\Phi}}_{\sigma}$ is partitioned as $\widetilde{\mathbf{\Phi}}_{\sigma}=\widetilde{\mathbf{\Phi}}_{\sigma}^{(\mathrm{full})}\cup \widetilde{\mathbf{\Phi}}_{\sigma}^{\mathrm{(com)}}$, where $\widetilde{\mathbf{\Phi}}_{\sigma}^{(\mathrm{full})}$ consists of \emph{full-support} \(d\)-tuples, and where \(|\widetilde{\mathbf{\Phi}}_\sigma^{\mathrm{(full)}}|=s^d\). Also recall (from \eqref{eq: phi_com}) that \(\widetilde{\mathbf{\Phi}}_\sigma^{\mathrm{(com)}} = \bigcup_{\beta=\lceil \frac{d}{s}\rceil}^{d-1}\bigcup_{\mathcal{I}\subset \sigma} \mathbf{C}_{\beta, \mathcal I, \sigma}\). Let us also remember that for any \(\beta \in[\lceil \frac{d}{s}\rceil,  d-1]\), each group \(\sigma \in {[f] \choose d}\) has exactly \({d \choose \beta}\) distinct subsets \(\mathcal I\) of cardinality \(\beta\), and that for each such \(\mathcal{I} \subset \sigma\), the cardinality of \(\mathbf{C}_{\beta, \mathcal I, \sigma}\) is either \(q_{\beta}+1\) or \(q_{\beta}\).
Now let us note that from \eqref{eq: q_beta}, we have 
    \begin{equation}
      \frac{t_\beta}{m_\beta}-1\le q_{\beta}= \lfloor \frac{t_\beta}{m_\beta}\rfloor\le \frac{t_\beta}{m_\beta} 
    \end{equation}
which gives 
 \begin{equation}
   \label{eq: lem3, ub, 2}
 |\widetilde{\mathbf{\Phi}}_{\sigma}^{\mathrm{(com)}}|\leq \sum_{\beta=\lceil\frac{d}{s}\rceil}^{d-1}(q_{\beta}+1)\binom{d}{\beta}\leq \sum_{\beta=\lceil\frac{d}{s}\rceil}^{d-1}\frac{t_\beta}{m_\beta}\cdot\binom{d}{\beta}+\sum_{\beta=\lceil\frac{d}{s}\rceil}^{d-1}\binom{d}{\beta}
\end{equation}
which in turn gives
 \begin{equation}
   \label{eq: lem3, lb, 2}
 |\widetilde{\mathbf{\Phi}}_{\sigma}^{\mathrm{(com)}}|\geq \sum_{\beta=\lceil\frac{d}{s}\rceil}^{d-1}q_{\beta}\binom{d}{\beta}\geq \sum_{\beta=\lceil\frac{d}{s}\rceil}^{d-1}\frac{t_\beta}{m_\beta}\cdot\binom{d}{\beta}-\sum_{\beta=\lceil\frac{d}{s}\rceil}^{d-1}\binom{d}{\beta}.
\end{equation}
Now using \eqref{eq: m_beta}, we get
    \begin{equation}
      \label{eq: lem3, ub, 3}
        \frac{t_\beta}{m_\beta} \cdot\binom{d}{\beta}= \frac{t_\beta}{{f-\beta \choose d-\beta}}\cdot\binom{d}{\beta}=t_\beta\cdot \frac{{f \choose \beta }}{{f\choose d}}  
    \end{equation}
noting also that  
\begin{equation}
\label{eq: lem3, ub, 4}
    \sum_{\beta=\lceil\frac{d}{s}\rceil}^{d-1}\binom{d}{\beta}\leq \sum_{\beta=0}^{d-1}\binom{d}{\beta}=2^d-d.
\end{equation}
At this point, using \eqref{eq: lem3, ub, 3} and \eqref{eq: lem3, ub, 4}, we bound \eqref{eq: lem3, ub, 2} and \eqref{eq: lem3, lb, 2}  as follows
  \begin{equation}
  \label{eq: lem3, ub, 5}
      |\widetilde{\mathbf{\Phi}}_{\sigma}^{\mathrm{(com)}}|\leq  \sum_{\beta=\lceil\frac{d}{s}\rceil}^{d-1}\frac{t_\beta \binom{f}{\beta}}{\binom{f}{d}}+2^d-d,
  \end{equation}
   \begin{equation}
  \label{eq: lem3, lb, 5}
      |\widetilde{\mathbf{\Phi}}_{\sigma}^{\mathrm{(com)}}|\geq  \sum_{\beta=\lceil\frac{d}{s}\rceil}^{d-1}\frac{t_\beta \binom{f}{\beta}}{\binom{f}{d}}-2^d+d.
  \end{equation}
Recalling~\eqref{eq: size of t_beta}, we have 
    \[\sum_{\beta=\lceil\frac{d}{s}\rceil}^{d-1}t_\beta \binom{f}{\beta}=\sum_{\beta=\lceil\frac{d}{s}\rceil}^{d}t_\beta \cdot \binom{f}{\beta} -t_d\cdot {f \choose d}={n \choose d}-t_d\cdot {f \choose d}\]
    which means that
    \begin{equation}
      \label{eq: lem3, ub, 6}
        |\widetilde{\mathbf{\Phi}}_{\sigma}^{\mathrm{(com)}}|\leq\frac{{n \choose d}-{f \choose d}t_d}{{f \choose d} }+2^d-d=\frac{{n \choose d}}{{f \choose d} }-t_d+2^d-d
    \end{equation}
    and that
    \begin{equation}
      \label{eq: lem3, lb, 6}
        |\widetilde{\mathbf{\Phi}}_{\sigma}^{\mathrm{(com)}}|\geq\frac{{n \choose d}-{f \choose d}t_d}{{f \choose d} }-2^d+d=\frac{{n \choose d}}{{f \choose d} }-t_d-2^d+d.
    \end{equation}
    Since \(t_d=s^d=|\widetilde{\mathbf{\Phi}}_{\sigma}^{\mathrm{(full)}}|\), we can lower and upper bound \(|\widetilde{\mathbf{\Phi}}_{\sigma}|\) as follows
     \begin{align}
      \label{eq: lem3, ub, 7}
        \frac{{n \choose d}}{{f \choose d} }-|\widetilde{\mathbf{\Phi}}_{\sigma}^{\mathrm{(full)}}|-2^d+d \le|\widetilde{\mathbf{\Phi}}_{\sigma}^{\mathrm{(com)}}|\leq\frac{{n \choose d}}{{f \choose d} }-|\widetilde{\mathbf{\Phi}}_{\sigma}^{\mathrm{(full)}}|+2^d-d,\\
       \frac{{n \choose d}}{{f \choose d} }-2^d+d \le  |\widetilde{\mathbf{\Phi}}_{\sigma}^{\mathrm{(com)}}|+|\widetilde{\mathbf{\Phi}}_{\sigma}^{\mathrm{(full)}}|\leq\frac{{n \choose d}}{{f \choose d} }+2^d-d,\\
       \frac{{n \choose d}}{{f \choose d} }-2^d+d \le  |\widetilde{\mathbf{\Phi}}_{\sigma}|\leq\frac{{n \choose d}}{{f \choose d} }+2^d-d
    \end{align}
    which concludes the proof.
    
\subsection{Proof of Lemma \ref{lem:lemma_pi_g=0}}
\label{Proofoflem: pi_g=0}
Let us first recall the IC design partition of \(\mathbf{A}_{n,d}\) from~\eqref{eq: A_n,d partition}.
We need to establish \(\pi =\max_{\sigma \in {[f] \choose d}}|\alpha(\widetilde{\mathbf{\Phi}}_\sigma)|\), recalling that  \(\widetilde{\mathbf{\Phi}}_\sigma=\widetilde{\mathbf{\Phi}}_\sigma^{\mathrm{(full)}}\cup \widetilde{\mathbf{\Phi}}_\sigma^{\mathrm{(com)}}\), where \(\widetilde{\mathbf{\Phi}}_\sigma^{\mathrm{(full)}}\cap \widetilde{\mathbf{\Phi}}_\sigma^{\mathrm{(com)}}=\emptyset\). We also recall from \eqref{eqfull} that  \(|\alpha(\widetilde{\mathbf{\Phi}}_\sigma^{\mathrm{(full)}})|=s\cdot d\) for every \(\sigma\in \binom{[f]}{d}\), since \(|\mathcal{F}_i|=s\), for all \(i\in [f]\).
Now, for any fixed \(\sigma\in \binom{[f]}{d}\), consider a \(d\)-tuple \(\mathbf{a}\in \widetilde{\mathbf{\Phi}}_\sigma^{\mathrm{(com)}} \) and let \(\mathcal{B}(\mathbf{a})=\mathcal{I}\) for some \(\mathcal{I}\in \binom{[f]}{\beta}\), where \(\beta\in[\lceil{\frac{d}{s}}\rceil, d-1]\). \textcolor{black}{Then, from \eqref{eq: phi_com}, we know} the support family of any \(d\)-tuple in \(\widetilde{\mathbf{\Phi}}_\sigma^{\mathrm{(com)}}\) is a subset of \(\sigma\). Consequently, we have \(\alpha(\widetilde{\mathbf{\Phi}}_\sigma^{\mathrm{(com)}})\subseteq \alpha(\widetilde{\mathbf{\Phi}}_\sigma^{\mathrm{(full)}})\) for every \(\sigma\in \binom{[f]}{d}\), which yields \(\alpha(\widetilde{\mathbf{\Phi}}_\sigma^{\mathrm{(full)}})=\alpha(\widetilde{\mathbf{\Phi}}_\sigma)\), which in turn means that 
\begin{equation}
\pi = \max\limits_{\sigma\in \binom{[f]}{d}}\alpha(\widetilde{\mathbf{\Phi}}_\sigma) = \max\limits_{\sigma\in \binom{[f]}{d}} \alpha(\widetilde{\mathbf{\Phi}}_\sigma^{(\mathrm{full})})=s\cdot d
\end{equation}
which concludes the proof.

\subsection{Proof of Lemma \ref{lem:lemma_pi_g>0}}
\label{Proofoflem: pi_g>0}
Let us recall the IC design, as described for Case~2 in Section~\ref{subsec: Case2}, and let us recall \(k\) from~\eqref{eq:largest_r} and \(f,g\) from~\eqref{eq:g,f}. 
Let us recall also that \(N' = \binom{f}{d}\), that \(n' = k \cdot s_0\) where \(s_0\) is set as in~\eqref{eq:s_0}, and that $\mathbf{A}_{n',d}$ was partitioned as
\begin{equation}
\mathbf{A}_{n',d} = \bigl\{\, 
\widetilde{\mathbf{\Phi}}_\sigma^{\mathrm{(full)}} 
\cup 
\widetilde{\mathbf{\Phi}}_\sigma^{\mathrm{(com)}}
\ \mid\ 
\sigma \in \tbinom{[f]}{d}
\,\bigr\}.   
\end{equation}
Finally, let us recall \textcolor{black}{from Lemma \ref{lem:lemma_pi_g=0}} that for every \(\sigma \in \binom{[f]}{d}\), it is the case that
\begin{equation}
\bigl|\alpha\bigl(
\widetilde{\mathbf{\Phi}}_\sigma^{\mathrm{(full)}} 
\cup 
\widetilde{\mathbf{\Phi}}_\sigma^{\mathrm{(com)}}
\bigr)\bigr|
= s_0 \cdot d.  
\end{equation}
We now analyze the excluded \(d\)-tuple set
\(
\mathbf{A}_{\mathrm{exc}} = \mathbf{A}_{n,d} \setminus \mathbf{A}_{n',d},
\)
noting that for each \(\sigma \in \binom{[f]}{d}\), we have the decomposition
\(
\widetilde{\mathbf{\Phi}}_\sigma
=
\widetilde{\mathbf{\Phi}}_\sigma^{\mathrm{(full)}}
\;\cup\;
\widetilde{\mathbf{\Phi}}_\sigma^{\mathrm{(com)}}
\;\cup\;
\widetilde{\mathbf{\Phi}}_\sigma^{\mathrm{(exc)}},
\)
with
\[
\bigl(
\widetilde{\mathbf{\Phi}}_\sigma^{\mathrm{(full)}} 
\cup 
\widetilde{\mathbf{\Phi}}_\sigma^{\mathrm{(com)}}
\bigr)
\cap 
\widetilde{\mathbf{\Phi}}_\sigma^{\mathrm{(exc)}}
= \emptyset.
\]
Consider any \(d\)-tuple \(\mathbf{a} \in \widetilde{\mathbf{\Phi}}_\sigma^{\mathrm{(exc)}}\). 
Let \(\mathcal{B}(\mathbf{a}) = \mathcal{I}\) for some 
\(\mathcal{I} \in \binom{[f]}{\beta}\), where 
\(\beta \in [\beta_{\mathrm{min}}, \beta_{\mathrm{max}}]\) 
(as defined in~\eqref{eq:beta_min}–\eqref{eq:beta_max}). 
For the excluded element set \(\mathcal{E}\) defined in~\eqref{eq:D}, we now have
\begin{equation}
\label{eq: a,E le g}
|\mathbf{a} \cap \mathcal{E}| = m \le |\mathcal{E}| = g.
\end{equation}
\textcolor{black}{ Using \eqref{eq: a,E le g} and the fact that \(\mathcal{I} \subset \sigma\) which we know from \eqref{eq: phi_exc},} we can conclude that the elements of any \(d\)-tuple in 
\(\widetilde{\mathbf{\Phi}}_\sigma^{\mathrm{(exc)}}\) lie in
\begin{equation}
\mathcal{F}_{\mathcal{I}} \cup (\mathbf{a} \cap \mathcal{E})
\subseteq \mathcal{F}_\sigma \cup \mathcal{E}    
\end{equation}
which means that \(\alpha\bigl(\widetilde{\mathbf{\Phi}}_\sigma^{\mathrm{(exc)}}\bigr)
\subseteq 
\mathcal{F}_\sigma \cup \mathcal{E}, \ 
\text{for all } 
\sigma \in \binom{[f]}{d}.
\) 
Since
\(
\alpha\bigl(
\widetilde{\mathbf{\Phi}}_\sigma^{\mathrm{(full)}}
\cup
\widetilde{\mathbf{\Phi}}_\sigma^{\mathrm{(com)}}
\bigr)
= \mathcal{F}_\sigma=\bigcup_{i \in \sigma} \mathcal{F}_i,
\)
we have
\begin{equation}
\alpha(\widetilde{\mathbf{\Phi}}_\sigma)
=
\alpha\!\left(
\widetilde{\mathbf{\Phi}}_\sigma^{\mathrm{(full)}}
\cup
\widetilde{\mathbf{\Phi}}_\sigma^{\mathrm{(com)}}
\cup
\widetilde{\mathbf{\Phi}}_\sigma^{\mathrm{(exc)}}
\right)
\subseteq
\alpha\!\left(
\widetilde{\mathbf{\Phi}}_\sigma^{\mathrm{(full)}}\right)\cup \alpha\!\left(
\widetilde{\mathbf{\Phi}}_\sigma^{\mathrm{(com)}}
\cup
\widetilde{\mathbf{\Phi}}_\sigma^{\mathrm{(exc)}}
\right)
\subseteq
\mathcal{F}_\sigma \cup \mathcal{E}   
\end{equation}
and taking the maximum over all \(\sigma\), we obtain
\begin{equation}
\pi
=
\max_{\sigma \in \binom{[f]}{d}}
\alpha(\widetilde{\mathbf{\Phi}}_\sigma)
\le
|\mathcal{F}_\sigma \cup \mathcal{E}|
= s_0 \cdot d + g    
\end{equation}
where the last equality follows from having
\(|\mathcal{F}_i| = s_0\) for all \(i \in [f]\), 
and from the fact that \(\mathcal{F}_\sigma\) contains \(d\) families.

\subsection{Proof of Lemma \ref{lem:ub-lb2} }
\label{proofoflem:ub2}
In Section~\ref{subsec: Case2}, given \(n\), \(d\), and \(N\), we can determine \(g\) and \(f\) from \eqref{eq:g,f}, and hence \(n' = n - g\). For each \(\sigma \in \binom{f}{d}\), recalling that the set \(\widetilde{\mathbf{\Phi}}_{\sigma}\) is partitioned into full-support \(d\)-tuples, complement \(d\)-tuples, and excluded \(d\)-tuples, we have
\begin{equation}
\label{eq:sum}
     |\widetilde{\mathbf{\Phi}}_{\sigma}|=|\widetilde{\mathbf{\Phi}}_\sigma^{\mathrm{(com)}}|+|\widetilde{\mathbf{\Phi}}_\sigma^{\mathrm{(full)}}|+|\widetilde{\mathbf{\Phi}}_\sigma^{\mathrm{(exc)}}|.
\end{equation}
Furthermore, from Lemma \ref{lem:ub-lb}, and given parameter set \(n', d\), and \(N'={f \choose d}\), we get that
\begin{equation}
\label{eq: 2 fisrtterm}
   \frac{{n' \choose d}}{{f \choose d}}-2^d+d \le |\widetilde{\mathbf{\Phi}}_\sigma^{\mathrm{(com)}}|+|\widetilde{\mathbf{\Phi}}_\sigma^{\mathrm{(full)}}|\leq \frac{{n' \choose d}}{{f \choose d}}+2^d-d.
\end{equation}
Recall now that \(\widetilde{\mathbf{\Phi}}_\sigma^{\mathrm{(exc)}} = \bigcup_{\beta=\beta_{\mathrm{min}}}^{\beta_{\mathrm{max}}}\bigcup_{\mathcal{I}\subset \sigma} \mathbf{R}_{\beta, \mathcal I, \sigma}\), and note that the method described in Appendix \ref{appendix: R_beta,I} guarantees that for any \(\beta \in \big[\beta_{\mathrm{min}}, \beta_{\mathrm{max}}]\),
\textcolor{black}{each \(\sigma \in {[f] \choose d}\) has exactly \({d \choose \beta}\) distinct subsets \(\mathcal I\) of cardinality \(\beta\), and that for each such \(\mathcal{I} \subset \sigma\), the cardinality of \(\mathbf{R}_{\beta, \mathcal I, \sigma}\) is either \(q_{\beta}+1\) or \(q_{\beta}\).}
This in turn means that 
\begin{equation}
\label{eq:phi_dis}
    |\widetilde{\mathbf{\Phi}}_{\sigma}^{\mathrm{(exc)}}|\leq \sum_{\beta=\beta_{\mathrm{min}}}^{\beta_{\mathrm{max}}}(q_{\beta}+1)\cdot\binom{d}{\beta}\\
    \leq\sum_{\beta=\beta_{\mathrm{min}}}^{\beta_{\mathrm{max}}}\frac{t_\beta}{m_\beta}\cdot\binom{d}{\beta}+ \sum_{\beta=\beta_{\mathrm{min}}}^{\beta_{\mathrm{max}}}\binom{d}{\beta}
\end{equation}
\begin{equation}
\label{eq:phi_dis2}
    |\widetilde{\mathbf{\Phi}}_{\sigma}^{\mathrm{(exc)}}|\geq \sum_{\beta=\beta_{\mathrm{min}}}^{\beta_{\mathrm{max}}}q_{\beta}\cdot\binom{d}{\beta},\\
    \geq\sum_{\beta=\beta_{\mathrm{min}}}^{\beta_{\mathrm{max}}}\frac{t_\beta}{m_\beta}\cdot\binom{d}{\beta}-\sum_{\beta=\beta_{\mathrm{min}}}^{\beta_{\mathrm{max}}}\binom{d}{\beta}.
\end{equation}
Now going back to \eqref{eq: m_beta}, we have that
\begin{equation}
\label{eq: sum q_beta}
   \frac{t_\beta}{m_\beta} \cdot\binom{d}{\beta}= \frac{t_\beta}{{f-\beta \choose d-\beta}}\cdot\binom{d}{\beta}=t_\beta\cdot \frac{{f \choose \beta }}{{f\choose d}}
\end{equation}
while also noting that  
\begin{equation}
\label{eq: sum d, beta}
    \sum_{\beta=\beta_{\mathrm{min}}}^{\beta_{\mathrm{max}}}\binom{d}{\beta}\leq \sum_{\beta=0}^{d-1}\binom{d}{\beta}=2^d-d.
\end{equation}
By substituting \eqref{eq: sum q_beta} and \eqref{eq: sum d, beta} in \eqref{eq:phi_dis}, we get the following inequalities
\begin{align}
\label{eq:lem5,phi_dis}
    \sum_{\beta=\beta_{\mathrm{min}}}^{\beta_{\mathrm{max}}}\frac{t_\beta \cdot\binom{f}{\beta}}{\binom{f}{d}}-2^d+d\le  |\widetilde{\mathbf{\Phi}}_{\sigma}^{\mathrm{(exc)}}|\leq \sum_{\beta=\beta_{\mathrm{min}}}^{\beta_{\mathrm{max}}}\frac{t_\beta \cdot\binom{f}{\beta}}{\binom{f}{d}}+2^d-d,
   \\
   \label{eq:lem5,phi_dis2}
   \frac{\sum_{\beta=\beta_{\mathrm{min}}}^{\beta_{\mathrm{max}}}|\mathbf{R}_{\beta}|}{N'}-2^d+d\le  |\widetilde{\mathbf{\Phi}}_{\sigma}^{\mathrm{(exc)}}|\leq\frac{\sum_{\beta=\beta_{\mathrm{min}}}^{\beta_{\mathrm{max}}}|\mathbf{R}_{\beta}|}{N'}+2^d-d
\end{align}
where the step from \eqref{eq:lem5,phi_dis} to \eqref{eq:lem5,phi_dis2} follows from \eqref{eq: R_beta, t_beta}. Furthermore, given that  \(\sum_{\beta=\beta_{\mathrm{min}}}^{\beta_{\mathrm{max}}}|\mathbf{R}_{\beta}|=|\mathbf{A_{\mathrm{exc}}}|={n \choose d}-{n' \choose d}\) we have that 
    \begin{equation}
        \label{eq: phi_diss}
        \frac{{n \choose d}-{n' \choose d}}{N'}-2^d+d\le |\widetilde{\mathbf{\Phi}}_{\sigma}^{\mathrm{(exc)}}|\leq \frac{{n \choose d}-{n' \choose d}}{N'}+2^d-d.
    \end{equation}
    We can now bound \eqref{eq:sum}, using \eqref{eq: 2 fisrtterm} and \eqref{eq: phi_diss}, in order to obtain
    \begin{equation}
    \label{eq:eqfinal1}
        |\widetilde{\mathbf{\Phi}}_{\sigma}|\leq \frac{{n-g \choose d}}{N'}+2^d-d+\frac{{n \choose d}-{n-g \choose d}}{N'}+2^d-d
    \end{equation}
    and 
     \begin{equation}
     \label{eq:eqfinal2}
        |\widetilde{\mathbf{\Phi}}_{\sigma}|\geq \frac{{n-g \choose d}}{N'}-2^d+d+\frac{{n \choose d}-{n-g \choose d}}{N'}-2^d+d
    \end{equation}
and since~\eqref{eq:eqfinal1} and \eqref{eq:eqfinal2} hold for every \(\sigma \in \binom{[f]}{d} \), we get  
\begin{equation}
    \label{phi_upper 1}
    \max\limits_{\sigma \in \binom{[f]}{d}} |\widetilde{\mathbf{\Phi}}_{\sigma}| \leq\frac{{n \choose d}}{N'}+2^{d+1}-2d,
\end{equation}
\begin{equation}
    \label{phi_upper 2}
    \min\limits_{\sigma \in \binom{[f]}{d}} |\widetilde{\mathbf{\Phi}}_{\sigma}| \geq\frac{{n \choose d}}{N'}-2^{d+1}+2d
\end{equation}
which concludes the proof.

\subsection{Proof of Lemma \ref{lem:N,N'}} \label{proofoflem: N, N'}

Directly from the definition of \(k\) in \eqref{eq:largest_r} and of \(N'\) in \eqref{eq:N'}, we note that
\begin{equation}\label{eq:k+1_choose_d}
    N<\binom{k+1}{d}, \quad N'=\binom{k}{d}
\end{equation}
which means that 
\begin{align}
    \frac{N}{N'}<\frac{\binom{k+1}{d}}{\binom{k}{d}}=\frac{k+1}{k+1-d}.
\end{align}
Since the function \(f(x) = \frac{x}{x - d}\) is decreasing in the domain \([d+1, \infty)\), we can conclude that 
\begin{equation}
    \max_{k \in \mathcal{K}_{\mathrm{valid}}} \frac{k+1}{k+1-d} \le d+1
\end{equation}
which means that 
\begin{equation}
    \frac{N}{N'} < d+1 \le 2^{d}.
\end{equation}

\subsection{Proof of Lemma \ref{lem: delta pi, A,n,d}}
\label{Proofoflem: delta pi, A,n,d}
Let us consider the design for Case~1, as described in Section~\ref{subsec: Case1}, where \(n = k \cdot s\) for some integer \(s\), with \(k\) defined in~\eqref{eq:largest_r}. Recall that in Lemma \ref{lem:ub-lb}, we have
\[
\max_{\sigma \in {[f] \choose d}}|\widetilde{\mathbf{\Phi}}_\sigma|
\le  \frac{\binom{n}{d}}{N'}+2^d-d.
\]
Recall also that, as discussed in Section~\ref{subsec: N' to N} --- when extending the partition from from \(N'\) to \(N\) groups --- each group is divided into \(q\) or \(p\) parts (see~\eqref{eq: q,p,r} for the definitions of \(q\), \(p\), and \(r\)). 
Now since \(q \le p\), we can conclude that 
\begin{align}
\label{eq: lemma 6,0}
  \max_{b \in [N]}|\widetilde{\mathbf{\Phi}}_b|
&\le\lceil \frac{\frac{\binom{n}{d}}{N'} +2^d -d}{q}\rceil \\
\label{eq: lemma 6,1}
&\le\lceil \frac{\binom{n}{d}}{q\cdot N'}\rceil  +\lceil\frac{2^d}{q} -\frac{d}{q}\rceil
\\
\label{eq: lemma 6,2}
&\le \frac{\binom{n}{d}}{q\cdot N'}+1  +\lceil\frac{2^d}{q}\rceil -\lfloor\frac{d}{q}\rfloor
\\
\label{eq: lemma 6,3}
&\le \frac{\binom{n}{d}}{N-r}+2^d  
\end{align}
where the transition from~\eqref{eq: lemma 6,0} to \eqref{eq: lemma 6,1} follows from the fact that \(\lceil x +y\rceil\le \lceil x\rceil+\lceil y\rceil\), the transition from~\eqref{eq: lemma 6,1} to \eqref{eq: lemma 6,2} follows from the fact that \(\lceil x \rceil\le x+1\) and \(\lceil x-y\rceil \le \lceil x\rceil-\lfloor y \rfloor\), while the transition from \eqref{eq: lemma 6,2} to \eqref{eq: lemma 6,3} follows from the fact that
\(1 \le q = \left\lfloor \frac{N}{N'} \right\rfloor \le d\) (cf.~Lemma~\ref{lem:N,N'}) 
and from the fact that \(N - r = N' \cdot q\). 

We can now proceed to obtain the following bound on \(\delta\)
\begin{align}
   \label{eq: th1: delta}
   \delta= \frac{\max_{b \in [N]}|\widetilde{\mathbf{\Phi}}_b|}{\lceil{n \choose d}/ N\rceil}\le\frac{\max_{b \in [N]}|\widetilde{\mathbf{\Phi}}_b|}{{n \choose d}/ N} &\le \frac{N}{N-r}+\frac{2^d\cdot N}{\binom{n}{d}}\\
   \label{eq: th1: delta2}
   &=\frac{N}{N'\cdot \lfloor\frac{N}{N'}\rfloor}+\frac{2^d\cdot N}{\binom{n}{d}}\\
   \label{eq: th1: delta3}
 &\le 2+\frac{2^d\cdot N}{\binom{n}{d}}
\end{align}
where the transition from~\eqref{eq: th1: delta} to \eqref{eq: th1: delta2} uses \eqref{eq: lemma 6,3} and also uses that
\(N = N' \cdot q + r\) where \(q = \left\lfloor \frac{N}{N'} \right\rfloor\), while the transition from~\eqref{eq: th1: delta2} to \eqref{eq: th1: delta3} follows from the fact that 
\(f(x) = \frac{x}{\lfloor x \rfloor} \le 2\) for \(x=\frac{N}{N'} \ge 1\).
Now, using the fact that \(N\le ( \frac{9}{10}\sqrt{\frac{n}{d}} )^d \le (\frac{n}{d})^{d/2}\) (cf.~Appendix~\ref{Appendix: limit on N}), and by substituting \({n \choose d}\ge (\frac{n}{d})^d\) and \(N\le (\frac{n}{d})^{d/2} \) in \eqref{eq: th1: delta3}, we get
\begin{equation}
\label{eq:th1: delta4}
\delta \; \le\; \frac{2^d N}{\binom{n}{d}}+2
\;\le\; \frac{2^d}{(\frac{n}{d})^{d/2}}+2.
\end{equation}
Let us now define the auxiliary variable
\begin{equation}
\label{eq: aux variable}
    Z \;\triangleq\; \frac{2^d}{(\frac{n}{d})^{d/2}}
\end{equation}
and easily note that if \(d\le \frac{n}{32}\), then \(Z\le 1\). 
Substituting now \(Z\le 1\) in \eqref{eq:th1: delta4}, we directly get 
\begin{equation}
  \delta \;\leq\;Z+2\;<\;3.
\end{equation}

Let us now shift attention to Case~2, as described in Section~\ref{subsec: Case2}, where the parameters \(s_0\) and \(g\) are defined as in \eqref{eq:s_0} and \eqref{eq:g,f} such that \(s_0 \mid n-g\). Also recall that \(n'=n-g\) and \(f=\frac{n'}{s}\), where the latter matches the value \(k\) assigned by \(\eqref{eq:largest_r}\). At this point we can also conclude from Lemma~\ref{lem:ub-lb2}, that
\[
\max\limits_{\sigma \in \binom{[f]}{d}} |\widetilde{\mathbf{\Phi}}_{\sigma}| \leq\frac{{n \choose d}}{N'}+2^{d+1}-2d.\]
Furthermore, after transitioning from \(N'\) to \(N\) groups, we also have
\begin{equation}
\label{eq: th1, maxmin, Sett2}
     \max_{b \in [N]}|\widetilde{\mathbf{\Phi}}_b|\leq \lceil\frac{\frac{{n \choose d}}{N'}+2^{d+1}-2d}{q}\rceil \le \frac{{n \choose d}}{N-r}+2^{d+1}
\end{equation}
and thus we have 
\begin{align}
 \label{eq: th1: Sett2,  delta}
 \delta\le\frac{\max_{b \in [N]}|\widetilde{\mathbf{\Phi}}_b|}{{n \choose d}/N}&\le \frac{N}{N-r}+\frac{2^{d+1}\cdot N}{\binom{n}{d}}\ \\
 \label{eq: th1: Sett2, delta2}
 &= \frac{N}{N'\cdot \lfloor\frac{N}{N'}\rfloor}+\frac{2^{d+1}\cdot N}{\binom{n}{d}}\\
\label{eq: th1: Sett2, delta3}
 &\le 2+\frac{2^{d+1}\cdot N}{\binom{n}{d}}
    \end{align}
where the transition from \eqref{eq: th1: Sett2,  delta} to \eqref{eq: th1: Sett2,  delta2} uses the fact that 
\(N = N' \cdot q + r\) with \(q = \left\lfloor \frac{N}{N'} \right\rfloor\), while the transition from \eqref{eq: th1: Sett2,  delta2} to \eqref{eq: th1: Sett2,  delta3} follows from the fact that
\(f(x) = \frac{x}{\lfloor x \rfloor} \le 2\), for \(x \ge 1\). 
Now note that --- similar to what we saw in \eqref{eq: aux variable} --- under the conditions 
\(N \le( \frac{9}{10}\sqrt{\frac{n}{d}})^d \le (\frac{n}{d})^{d/2}\) and \(d \le \frac{n}{32}\), we have \(Z \le 1\), which thus yields
\begin{equation}
    \label{eq: th1: Sett2, delta5}
    \delta\le 2+2Z\le 4.
\end{equation}

Now, regarding the communication cost $\pi$, first under Case 1, we compute the gap between the achievable \(\pi\) and the optimal communication cost \(\pi^\star\), by applying the binomial approximation on \(N'=\binom{k}{d}=\binom{{n}/{s}}{d}\), to obtain
\begin{equation}
    \left(\frac{n}{s d}\right)^d \leq N'\leq \left(\frac{e n}{s d}\right)^d
\end{equation}
which then gives 
\begin{equation}
    \frac{n}{N'^{\frac{1}{d}}}\leq s\cdot d \leq \frac{e n}{N'^{\frac{1}{d}}}.
\end{equation}
Furthermore, from  Lemma \ref{lem:lemma_pi_g=0}, we have 
\begin{equation}
\label{eq: th1, pi1}
  \pi =s\cdot d \leq \frac{e n}{N'^{\frac{1}{d}}}
\end{equation}
and from Appendix \ref{proofoflem:fundlower} and \eqref{eq: th1, pi1}, we have the gap
\begin{align}
    \frac{\pi}{\pi^\star}\leq \frac{{e n}/{N'^{\frac{1}{d}}}}{{ n}/{N^{\frac{1}{d}}}}\leq e\left(\frac{N}{N'}\right)^{\frac{1}{d}}.
\end{align}
Now using Lemma \ref{lem:N,N'}, we can conclude that
\begin{equation}
   \frac{\pi}{\pi^\star}\leq 2 e.
\end{equation}
Now, let us calculate the same gap for Case 2 (Section~\ref{subsec: Case2}). First, from Lemma \ref{lem:lemma_pi_g>0}, we have the following achievable communication cost
    \begin{equation}
    \label{eq: th1, pi, Sett2,1}
        \pi=  s_0 \cdot d +g\le 2s_0\cdot d
    \end{equation}
and \textcolor{black}{since \(\binom{n}{k}\le (\frac{en}{k})^{k}\)}, by applying the binomial approximation on \(N'={ {n'}/{s_0} \choose d}\), we get 
\begin{equation}
\label{eq: th1, pi, Sett2,2}
    (\frac{n'}{s_0\cdot d})^d \leq N'\leq (\frac{e\cdot n'}{s_0\cdot d})^d
\end{equation}
which yields 
\begin{equation}
\label{eq: th1, pi, Sett2,3}
    \frac{n'}{N'^{\frac{1}{d}}}\leq s_0\cdot d \leq \frac{e\cdot n'}{N'^{\frac{1}{d}}}
\end{equation}
and thus from \eqref{eq: th1, pi, Sett2,1}, we get
\begin{equation}
\label{eq: th1, pi, Sett2,4}
  \pi\leq 2s_0\cdot d \leq \frac{2e\cdot n'}{N'^{\frac{1}{d}}}.  
\end{equation}
Finally, using the result from Appendix~\ref{proofoflem:fundlower}, we obtain 
\begin{align}
\label{eq: th1, Sett2, mult1}
 \frac{\pi}{\pi^\star}&\leq \frac{2e\cdot n'/N'^{\frac{1}{d}}}{n/N^{\frac{1}{d}}}=2e\cdot \frac{n'}{n}\cdot (\frac{N}{N'})^{\frac{1}{d}}\\
  \label{eq: th1, Sett2, mult2}
  &\leq 2e\cdot (\frac{N}{N'})^{\frac{1}{d}}\\
  \label{eq: th1, Sett2, mult3}
  &\leq 4\cdot e 
\end{align}
where the transition from \eqref{eq: th1, Sett2, mult1} to \eqref{eq: th1, Sett2, mult2} follows from the fact that \(n' \le n\), while the transition from \eqref{eq: th1, Sett2, mult2} to \eqref{eq: th1, Sett2, mult3} follows from Lemma~\ref{lem:N,N'}.

\subsection{Proof of Lemma \ref{lem: conc on group size}}
\label{proofoflem: conc on group size}
Let us recall the partition of \(\mathbf{A}_{n,d}\) into \(N\) groups \(\widetilde{\mathbf{\Phi}}_1, \widetilde{\mathbf{\Phi}}_2, \ldots, \widetilde{\mathbf{\Phi}}_N\) seen in Sections \ref{subsec: Case1}, \ref{subsec: Case2}, and \ref{subsec: N' to N}, and let us first define
\begin{equation}
\label{eq: random, m_min, m_max}
m_b\triangleq|\widetilde{\mathbf{\Phi}}_b|, \ b\in[N], \ \ \ \     m_{\min} \triangleq \min_{b\in[N]}m_b, \quad m_{\max} \triangleq \max_{b\in[N]}m_b.
\end{equation}
Let us also recall that the balance parameter (computation cost) first takes the form
\begin{equation}
\label{eq: balance parameter}
 \delta=\delta_{\mathbf{A}_{n,d}} = \frac{m_{\mathrm{max}}}{\lceil{n \choose d}/N\rceil}
\end{equation}
where from Lemma \ref{lem: delta pi, A,n,d}, we know that \(\delta_{\mathbf{A}_{n,d}}\le 4\). Given \(\mathbf{X}\) and its corresponding partition into $\mathbf{\Phi}_1,\dots,\mathbf{\Phi}_N$, where $\mathbf{\Phi}_b=\widetilde{\mathbf{\Phi}}_b\cap \mathbf{X}, \ b\in[N]$, let us upper bound the corresponding \[\delta_{\mathbf{X}} =  \frac{\max_{b\in[N]} \mathbf{\Phi}_b}{\lceil{|\mathbf{X}|}/N\rceil}.\]
To do so, for each \(d\)-tuple $\mathbf{t} \in \widetilde{\mathbf{\Phi}}_b$, let us define $I_\mathbf{t} $ to be the indicator random variable such that $I_\mathbf{t} = 1$ when $\mathbf{t} \in \mathbf{X}$, and $I_\mathbf{t} = 0$ otherwise. Recall that we proceed under the assumption that $\mathbf{X}$ is a result of random thinning from $\mathbf{A}_{n,d}$. Since this random thinning acts independently on each \(d\)-tuple, the cardinality of $\mathbf{\Phi}_b$, takes the form
\begin{equation}
\label{eq: lem 11, 1}
  \mu_b = |\mathbf{\Phi}_b|= \sum_{\mathbf{t}  \in \widetilde{\mathbf{\Phi}}_b} I_\mathbf{t}
\end{equation}
where each $I_\mathbf{t}$ follows a Bernoulli distribution with probability $\varphi$, corresponding to having $\mathbb{E}[I_\mathbf{t}] = \varphi$. This in turn means that 
\begin{equation}
\label{eq: lem 11,2}
   \mathbb{E}[\mu_b] = \mathbb{E}\big[\sum_{ \mathbf{t} \in \widetilde{\mathbf{\Phi}}_b} I_\mathbf{t}\big] = \sum_{\mathbf{t} \in \widetilde{\mathbf{\Phi}}_b} \mathbb{E}[I_\mathbf{t}] = \varphi \cdot |\widetilde{\mathbf{\Phi}}_b| =\varphi \cdot m_b.
\end{equation}
Using Corollary~4.6 in \cite{mitzenmacher2017probability}, for \(0 < \varepsilon < 1\) and \(b \in [N]\), we have
\begin{align}
\Prob\big[|\mu_b-\mathbb{E}(\mu_b)| \ge\varepsilon\mathbb{E}(\mu_b)\big] \le 2\exp\!\big(-\tfrac{\varepsilon^2\mathbb{E}(\mu_b)}{3}\big)
\end{align}
and then, by applying the union bound over all \(N\) groups, we get
\begin{align}
 \Prob\big[\exists b\mid \mu_b\not\in[(1-\varepsilon)\mathbb{E}(\mu_b),(1+\varepsilon)\mathbb{E}(\mu_b)]\big]  \leq 2\ N\cdot\max_{b \in [N]}\big(\exp( - \frac{\epsilon^2\mathbb{E}(\mu_b)}{3} )\big)\\
 \label{eq: lem 11, 5}
  \leq 2\ N\cdot\exp\big( - \frac{\epsilon^2}{3}\cdot\min\limits_{b \in [N]}\mathbb{E}(\mu_b)\big ).
\end{align}
Let us now recall that \(\min\limits_{b \in [N]}\mathbb{E}[ \mu_{b} ]=\varphi\cdot m_{\mathrm{min}}\). and let us choose \(\epsilon = \frac{1}{4}\) and \(\eta = \frac{1}{n}\), and substitute these values into \eqref{eq: lem 11, 5}. By solving the following inequality 
\[
2\ N\cdot\exp( - \frac{\varphi\cdot m_{\mathrm{min}}}{48})\le \eta=\frac{1}{n}
\]
we can conclude that under the condition
\begin{equation}
   \label{eq:mu-condition 4}
\varphi\cdot m_{\mathrm{min}}\;\ge\; 48\big(\log(2N)+\log(n)\big) 
\end{equation}
it is the case that 
\begin{align}
\label{eq: lem 11, 6}
     \Pr\Big[\exists b\mid \mu_b\not\in[\frac{3}{4}\mathbb{E}(\mu_b),\frac{5}{4}\mathbb{E}(\mu_b)]\Big] \leq 2\ N\cdot\exp\left( - \frac{\varphi\cdot m_{\mathrm{min}}}{48} \right)\le \frac{1}{n}
\end{align}
which means that, with probability at least \(1-\frac{1}{n}\), we have 
\begin{align}
    \mu_b = \varphi m_b \pm \frac{1}{4}\varphi m_{b}, \\
    |\mathbf{\Phi}_b|=\varphi  |\widetilde{\mathbf{\Phi}}_b|  \pm \frac{1}{4}\varphi |\widetilde{\mathbf{\Phi}}_b|
\end{align}
which means that 
\begin{equation}
    \max_{b \in [N]} |\mathbf{\Phi}_b|\le \frac{5}{4}\varphi\cdot\max_{b \in [N]} |\widetilde{\mathbf{\Phi}}_b|
\end{equation}
which in turn yields
\begin{equation}
 \delta_{\mathbf{X}}= \frac{\max_{b \in [N]} |\mathbf{\Phi}_b|}{\varphi. \binom{n}{d}}\le \frac{ \frac{5}{4}\varphi\cdot\max_{b \in [N]} |\widetilde{\mathbf{\Phi}}_b|}{\varphi. \binom{n}{d}}=\frac{5}{4} \delta=5
 \end{equation}
which completes the proof.

\subsubsection{Simplifying the condition in~\eqref{eq:mu-condition 4}}
Let us now simplify the condition \eqref{eq:mu-condition 4}. To do so, let us recall that from Lemma \ref{lem:ub-lb2}, we have
\[
\min\limits_{\sigma \in \binom{[f]}{d}} |\widetilde{\mathbf{\Phi}}_{\sigma}| \geq\frac{{n \choose d}}{N'}-2^{d+1}+2d.\]
Furthermore, let us recall that in Section \ref{subsec: N' to N} (where we extend the partition, from \(N'\) to \(N\) groups), each group is divided into \(q\) or \(p\) parts (see~\eqref{eq: q,p,r} for the definitions of \(q\), \(p\), and \(r\)), and since 
\(q \le p\), we have  
\begin{align}
\label{eq: lem,m_min,0}
   m_{\mathrm{min}}= \min\limits_{b\in [N]} |\widetilde{\mathbf{\Phi}}_{b}| \ge \lfloor\frac{\frac{{n \choose d}}{N'}-2^{d+1}+2d}{p}\rfloor\ge \frac{{n \choose d}}{pN'}-\frac{2^{d+1}}{p}+\frac{2d}{p}-1\\
   \label{eq: lem,m_min,1}
   \ge \frac{{n \choose d}}{N+N'}-\frac{2^{d+1}}{p}+\frac{2d}{p}-1\\
   \label{eq: lem,m_min,2}
   \ge \frac{{n \choose d}}{2N}-2^{d+1} 
\end{align}
where the transition from~\eqref{eq: lem,m_min,0}  to \eqref{eq: lem,m_min,1} follows from the facts that a) if \(r=0\), then \(pN'=N\), and b) if \(r\neq0\), then \(p=q+1\) and thus \(N\le pN'=N+N'-r\le N+N'\le 2N\). Furthermore, in the above, the transition~\eqref{eq: lem,m_min,1} to \eqref{eq: lem,m_min,2} follows from the fact that \(1\le p\le d+1\le 2d\).
At this point, let us substitute \eqref{eq: lem,m_min,2} in condition \eqref{eq:mu-condition 4}, and let us also denote the corresponding threshold by 
\begin{equation}\label{eq:varphi_min}
\varphi_{\mathrm{min}}
\triangleq \frac{96\,N\,\log(2Nn)}{\binom{n}{d} - 2^{\,d+2}N}.
\end{equation}
Let us now note that from Appendix~\ref{Appendix: limit on N}, we have that \(N \le (\frac{9}{10}\sqrt{\frac{n}{d}})^d\le(\frac{n}{d})^{d/2} \), and let us note that 
for a fixed \(d \ge 2\) and \(n\ge 32 d\), the denominator of \eqref{eq:varphi_min} is larger than \(\frac{1}{2}(n/d)^{d}\). 
Now, since
\begin{equation}
\label{eq: ph_min, denom,inq}
 \binom{n}{d}-2^{\,d+2}N
\ge \left(\frac{n}{d}\right)^{d}- 2^{\,d+2} (\frac{n}{d})^{d/2}   
\end{equation}
when \(n\ge 32 d\) and \(d\ge 2\), we can conclude that 
\begin{equation}
\label{eq: phi_min, denominator,0}
\left(\frac{n}{d}\right)^{d}- 2^{\,d+2} (\frac{n}{d})^{d/2}= \left(\frac{n}{d}\right)^{d}\big(1-\frac{2^{d+2}}{(\frac{n} {d})^{d/2}}\big)\ge\left(\frac{n}{d}\right)^{d}(1-2^{d+2-2.5d})\ge \frac{1}{2}(n/d)^{d}.
\end{equation}
Substituting \eqref{eq: phi_min, denominator,0} and \(N \le (\frac{n}{d})^{d/2}\) in \eqref{eq:varphi_min} now yields the simpler threshold 
\begin{align}
\label{eq: lem 11, 7}
  \varphi_{\mathrm{min}}=\frac{96\,N\,\log(2Nn)}{\binom{n}{d} - 2^{\,d+2}N}\le \frac{96\, \log(2 (\frac{n}{d})^{d/2} \cdot n)}{\frac{1}{2}(n/d)^{d/2}}\\
  \label{eq: lem 11, 8}
  \le \frac{192\, \log(n^{d/2} \cdot n)}{(n/d)^{d/2}}
  =\frac{192\, d^{d/2} (1+\frac{d}{2})\log(n)}{n^{d/2}}\\
  \label{eq: lem 11, 9}
  \le 192\, (ed)^{d/2} \cdot\frac{\log(n)}{n^{d/2}}
\end{align}
\textcolor{black}{where the transition from \eqref{eq: lem 11, 7} to \eqref{eq: lem 11, 8} follows from the fact that \(2 (\frac{1}{d})^{d/2}\le1\) for \(d\ge 2\), while the step from \eqref{eq: lem 11, 8} to \eqref{eq: lem 11, 9} follows from the fact that \(1+x\le e^x\), for \(x\ge0\).} This means that, for any given \(d \ge 2\), we have the simplified condition  
\begin{equation}
 \varphi_{\mathrm{min}}
\le
C_d\,\frac{\log n}{n^{d/2}}   
\end{equation}

where the value \(C_d=192\, (ed)^{d/2}\) depends only on \(d\). In conclusion, we have \(\delta_{\mathbf{X}}\le 5\) with probability at least \(1-\frac{1}{n}\) under the condition \(\varphi\ge C_d\,\frac{\log n}{n^{d/2}}\).

\section{Various Proofs}
\subsection{Proof of Lower Bound on \(\pi^\star\)}
\label{proofoflem:fundlower}
Let us consider an arbitrary $N$-group partition \(\mathbb{P}=\{\mathbf{\Phi}_1, \mathbf{\Phi}_2, \ldots, \mathbf{\Phi}_N \}\) of $\mathbf{X} \subseteq\mathbf{A}_{n,d}$.
For $\pi_b \triangleq |\alpha(\mathbf{\Phi}_b)|$, it is easy to see that 
\begin{equation}
\label{eq: lem fundlower 1}
    \sum_{b=1}^N \binom{\pi_b}{d} \geq |\mathbf{X}|.
\end{equation}
Furthermore, for $\pi_{\mathrm{max}} \triangleq \max_b \pi_b$, and since $\binom{\pi_{\mathrm{max}}}{d}\ge \binom{\pi_b}{d}$, we can conclude that 
\begin{equation}
\label{eq: lem fundlower 2}
   N\binom{\pi_{\mathrm{max}}}{d} \ge  |\mathbf{X}|=\varphi\cdot {n \choose d}
\end{equation}
which gives
\begin{equation}
    \label{eq: lem fundlower 3}
    \frac{ \pi_{\mathrm{max}}!}{(\pi_{\mathrm{max}}-d)!d!} \geq \frac{\varphi \cdot n!}{N\cdot (n-d)! \cdot d!}
\end{equation}
which --- after expanding the factorials --- yields
\begin{equation}
\label{eq: lem fundlower 4}
\frac{ \pi_{\mathrm{max}}\cdot(\pi_{\mathrm{max}}-1)\cdots (\pi_{\mathrm{max}}-d+1)}{n\cdot (n-1)\cdots (n-d+1)} \geq \frac{\varphi}{N}.
\end{equation}
Furthermore, as \(d \leq \pi_{\mathrm{max}}\leq n\), we have 
\begin{align}
\label{eq: lem fundlower 5}
  \frac{\pi_{\mathrm{max}}}{n}\geq\frac{\pi_{\mathrm{max}}-k}{n-k}, \quad \text {for \(1\leq k \leq d-1\)}
\end{align}
which combines with \eqref{eq: lem fundlower 4}, \eqref{eq: lem fundlower 5}, to give % loose lower bound, when d is big compare to n, Entropy bound
\begin{equation}
\frac{(\pi_{\mathrm{max}})^d}{n^d}\geq \frac{\varphi}{N}
\end{equation}
which translates to $
\pi_{\mathrm{max}}\geq \frac{\varphi^\frac{1}{d}\cdot n }{N^{\frac{1}{d}}}$, which gives \begin{equation}
\pi^{\star}\geq \frac{\varphi^\frac{1}{d}\cdot n }{N^{\frac{1}{d}}}.
\end{equation}

\subsection{Deriving the Cardinality of \(\mathbf{C}_{\beta, \mathcal{I}}\)}
\label{proofoflem:count}
We here derive the cardinality of $\mathbf{C}_{\beta,\mathcal{I}} = \left\{ \mathbf{a} \in \mathbf{C}_\beta \;\middle|\; \mathcal{B}(\mathbf{a})=\mathcal{I}\right\}$ from \eqref{eq:C_beta, I}, showing that \( |\mathbf{C}_{\beta,\mathcal{I}} | \) remains fixed for all \( \mathcal{I} \in {[f] \choose \beta} \). 

First, from Section~\ref{subsec: Case1}, we recall that we have \(f\) disjoint families \(\mathcal{F}_i\subset [n]\), \(i \in [f]\), each having cardinality \(s\).
We also recall from~\eqref{eq:C_beta} that for each \(\beta \in [\lceil\frac{d}{s}\rceil, d]\), the set \(\mathbf{C}_\beta\) is the collection of \(d\)-tuples that span exactly \(\beta\) families.

Let us now remember the well known result (cf.~\cite{stanley2011enumerative}) that the coefficient of term \(x^k\) of the generating function 
\begin{equation}
\label{eq: lem2, gen, 1}
    G_i(x)=(1 + x)^s
\end{equation}
is equal to the number of ways one can select \(k\) elements (\(k\)-tuples) from $\mathcal{F}_i$, where \(|\mathcal{F}_i|=s\).

To calculate \(|\mathbf{C}_\beta|\), we count all selections of \(d\)-tuples that \textcolor{black}{intersect} exactly \(\beta\) families. First, we choose which \(\beta\) families are used, giving us \(\binom{f}{\beta}\) choices. For the family \(\mathcal{F}_i\), the generating function for choosing any \(d\)-tuples of that family is given in \eqref{eq: lem2, gen, 1}. %To ensure that the selection intersects \(\mathcal{F}_i\)
We remove the empty choice from \eqref{eq: lem2, gen, 1} to ensure that we pick at least one element from \(\mathcal{F}_i\), and thus we replace \((1+x)^s\) by \((1+x)^s - 1\), which means that the number of ways to choose exactly \(d\) elements (\(d\)-tuples) from those \(\beta\) families is the coefficient of \(x^d\) in \(\bigl((1+x)^s-1\bigr)^\beta\). Call this coefficient $z_{\beta,d}$, and see that   
\begin{equation}
\label{eq: lem2, gen, 3}
|\mathbf{C}_\beta| \;=\; \binom{f}{\beta}z_{\beta,d}.
\end{equation}
Let \(A=(1+x)^s\), and recall from the binomial theorem that
\begin{align}
\bigl((1+x)^s-1\bigr)^\beta
&= (A-1)^\beta = \sum_{i=0}^{\beta} \binom{\beta}{i}A^{i}(-1)^{\beta-i}
\end{align}
which --- after substituting for \(A\) --- gives
\begin{align}
\label{eq: lem2, gen, 4}
\bigl((1+x)^s-1\bigr)^\beta
&=\sum_{i=0}^{\beta} (-1)^{\beta-i}\binom{\beta}{i}(1+x)^{s \cdot i}.
\end{align}
Note now that combining~\eqref{eq: lem2, gen, 3} and \eqref{eq: lem2, gen, 4}, for every \(\lceil \frac{d}{s}\rceil \le \beta \le d\), gives
\begin{equation}
\label{eq: size of C_beta}
    |\mathbf{C}_\beta| \;=\; \binom{f}{\beta}\sum_{i=0}^{\beta}(-1)^{\beta-i}\binom{\beta}{i}\binom{s \cdot i}{d}.  
\end{equation}
Now let us recall that the different sets \(\mathbf{C}_\beta\) are disjoint and cover the entire \(\mathbf{A}_{n,d}\), which thus means that
\begin{equation}
\label{eq: Identity A_n,d and C_beta}
  \binom{n}{d} = \sum_{\beta=\lceil\frac{d}{s}\rceil}^{d}\binom{f}{\beta}\sum_{i=0}^{\beta}(-1)^{\beta-i}\binom{\beta}{i}\binom{s \cdot i}{d}.  
\end{equation}
Let us now recall that for a fixed \(\mathcal{I} \in {[f] \choose \beta}\), the set \( \mathbf{C}_{\beta,\mathcal{I}} \) contains all \( d \)-tuples in \(\mathbf{C}_\beta\) supported by all \( \beta \) families in \( \mathcal{I} \). Since though the families are disjoint and correspond to the same \( s \), their cardinality \( |\mathbf{C}_{\beta,\mathcal{I}} | \) is the same for all \( \mathcal{I} \in {[f] \choose \beta} \). Denoting  \( t_\beta\triangleq|\mathbf{C}_{\beta,\mathcal{I}} | \), we have
\begin{align}
    \label{eq: C_beta, c_beta, I}
    |\mathbf{C}_\beta|=\sum_{\mathcal{I} \in {[f] \choose \beta}}|\mathbf{C}_{\beta, \mathcal I}|={f \choose \beta}\cdot t_\beta
    \end{align}
    which gives 
    \begin{align}
     \label{eq: size of t_beta}
    t_\beta&=|\mathbf{C}_\beta|/{f \choose \beta} 
    =\sum_{i=0}^{\beta}(-1)^{\beta-i}\binom{\beta}{i}\binom{s\cdot i}{d}. %\label{eq: exact size of t_beta, C_beta,I}
\end{align}

\subsection{Explicit Description of Set \(\mathbf{C}_{\beta, \mathcal{I}, \sigma}\)}
\label{appendix: C_beta,I}
We fix a group \(\sigma \in {[f] \choose d}\). For each \(\beta \in [\lceil\frac{d}{s}\rceil, d-1]\), the group \(\sigma\) can generate exactly
\(
\binom{d}{\beta}
\)
distinct subsets \(\mathcal{I}\in {[f] \choose \beta}\) of cardinality \(|\mathcal I|=\beta\).
For each such $\mathcal I$, let us consider the set \(\mathbf{C}_{\beta,\mathcal I}\) to be the set of all \(t_\beta\) distinct \(d\)-tuples from \(\mathbf{C}_\beta\) which intersect with \(\mathcal{I}\). For  a \(d\)-tuple \(\mathbf{a}\), we say that \(\mathbf{a}\) is \emph{eligible} for assignment to a group \(\widetilde{\mathbf{\Phi}}_\sigma\) if \(\mathcal{B}(\mathbf{a})\subset \sigma\), recalling also that for each \(\mathbf{a}\in \mathbf{C}_\beta\), there is one \(\mathcal I \in {[f] \choose \beta}\) such that \(\mathcal{I}= \mathcal{B}(\mathbf{a})\). Let us define the set of eligible groups for each \(\mathcal I \in {[f] \choose \beta}\) as
\begin{equation}
\label{eq: g_beta,I}
    \mathcal{G}_{\beta, \mathcal{I}}\triangleq \{\sigma \in {[f] \choose d}| \ \mathcal{I}\subseteq \sigma\}
    \end{equation}
and let us also denote the number of eligible groups as
\begin{equation}
\label{eq: m_beta}
  m_\beta \;\triangleq\; |\mathcal{G}_{\beta, \mathcal{I}}|=\binom{f-\beta}{d-\beta}.  
\end{equation}  
Now recall that \(t_\beta=|\mathbf{C}_{\beta, \mathcal I}|\) (Appendix \ref{appendix: size of t_beta}), and let us denote
\begin{equation}
 \label{eq: q_beta}
    q_\beta \;\triangleq\; \Big\lfloor\frac{t_\beta}{m_\beta}\Big\rfloor
\end{equation}
and %the remainder 
\begin{equation}
 \label{eq: r_beta}
    r_\beta \;\triangleq\; t_\beta - q_\beta \cdot m_\beta, \qquad 0\le r_\beta < m_\beta.
\end{equation}
Let us now consider \(\mathbf{C}_{\beta,\mathcal I}\) in lexicographic order and let us denote the ordered list by
\[
\mathbf{C}^{\mathrm{lex}}_{\beta,\mathcal I}
\;\triangleq\;
\big( \mathbf{c}_{\beta,\mathcal I}^{(1)}, \mathbf{c}_{\beta,\mathcal I}^{(2)}, \dots, \mathbf{c}_{\beta,\mathcal I}^{(t_\beta)}\big)
\]
where \(\mathbf{c}_{\beta,\mathcal I}^{(\ell)}\) describes the \(\ell\)-th \(d\)-tuple in this lexicographic ordering.
Subsequently, let us fix the lexicographic ordering of the \(m_\beta\) groups in \(\mathcal{G}_{\beta,\mathcal I}\), and let us denote this with 
\[
\mathcal{G}_{\beta,\mathcal I}^{\mathrm{lex}}
\;=\;
\big(\sigma_1,\sigma_2,\dots,\sigma_{m_\beta}\big).
\]
We now partition the \(\mathbf{C}^{\mathrm{lex}}_{\beta,\mathcal I}\) among the \(m_\beta\) groups \(\sigma_1,\dots,\sigma_{m_\beta}\).  
Since \(q_\beta=\lfloor \frac{t_\beta}{m_\beta}\rfloor\), each group receives either \(q_\beta\) or \(q_\beta+1\) \(d\)-tuples, which will mean that, in particular, \(r_\beta\) groups receive \(q_\beta+1\) \(d\)-tuples and the remaining \(m_\beta-r_\beta\) groups receive \(q_\beta\) \(d\)-tuples.
Let us now define the index of \(\sigma\) (given a specific \(\mathcal I \subset \sigma\)) by
\[
J_{\beta,\mathcal I}(\sigma)
\;\triangleq\; \{ j \ \  \text{s.t. }\ \sigma=\sigma_j\text{ for some }j\in [m_\beta], \sigma_j \in \mathcal{G}_{\beta,\mathcal I}^{\mathrm{lex}} \}.
\]  
For \(\mathcal I \subset \sigma\) with \(J_{\beta,\mathcal I}(\sigma)=j\in\{1,\dots,m_\beta\}\), we now define the starting and ending indices of the block allocated to \(\sigma\) by
\begin{align}
s_{j,\beta,\mathcal I} &\;=\; (j-1)q_\beta + \min(j-1,r_\beta) + 1,\label{eq:start 1}\\[4pt]
e_{j,\beta,\mathcal I} &\;=\; j q_\beta + \min(j,r_\beta)\label{eq:end 2}
\end{align}
which gives us the final set 
\begin{equation}\label{eq:allocated-set,C}
\quad
\mathbf{C}_{\beta,\mathcal I,\sigma}
\;=\;
\big\{\, \mathbf{c}_{\beta,\mathcal I}^{(\ell)} \;\mid\; \ell = s_{j,\beta,\mathcal I},\; s_{j,\beta,\mathcal I}+1,\; \dots,\; e_{j,\beta,\mathcal I}
\,\big\}
\quad
\end{equation}
of allocated \(d\)-tuples to group \(\sigma\) for this \(\mathcal I\).

\subsection{Deriving the Cardinality of \(\mathbf{R}_{\beta, \mathcal{I}}\)}
\label{appendix: size of t_beta}
As we recall from Section~\ref{subsec: Case2} (right above Eq.~\eqref{eq: phi_exc}), for any \(\mathcal{I}\in\binom{[f]}{\beta}\), \(\mathbf{R}_{\beta,\mathcal{I}}\) is the collection of all excluded \(d\)-tuples that are supported by exactly the \(\beta\) families in \(\mathcal{I}\).
Since the families are disjoint and all have size \(s_0\), for all distinct \(\mathcal{I},\mathcal{I}'\in {[f] \choose \beta} \), we have
\[
\mathbf{R}_{\beta,\mathcal{I}} \cap \mathbf{R}_{\beta,\mathcal{I}'}=\emptyset
\]
and thus we have 
\begin{equation}
    \label{eq: R_beta, t_beta}
    \mathbf{R}_{\beta} \;=\; \bigcup_{\mathcal{I}\in\binom{[f]}{\beta}}\mathbf{R}_{\beta,\mathcal{I}}
\qquad\Longrightarrow\qquad
\big|\mathbf{R}_{\beta}\big| \;=\; \binom{f}{\beta}\cdot\,|\mathbf{R}_{\beta,\mathcal{I}}|.
\end{equation}
Also, the cardinality \(|\mathbf{R}_{\beta,\mathcal{I}}|\) is the same for any choice of \(\mathcal{I}\) ; we denote this cardinality by \(t_{\beta}\) and note that
\begin{equation}
    t_{\beta} \;\triangleq\; |\mathbf{R}_{\beta,\mathcal{I}}|,\quad\text{for each }\mathcal{I}\in\binom{[f]}{\beta}.
\end{equation}
We also know that %now compute \(|\mathbf{R}_{\beta}|\) as follows.
        \[
        |\mathbf{R}_{\beta}|=\sum_{m=1}^{min\{d-\beta, g\}}|\mathbf{R}_{m,\beta}|.
        \]
        We then use the format of $\mathbf{R}_{m,\beta}$ in \eqref{eq:R_m,beta}, as well as the exclusion and inclusion principle and the generating function approach (cf.~\cite{stanley2011enumerative}), to get 
        \[
        |\mathbf{R}_{\beta}|=\sum_{m=1}^{min\{d-\beta, g\}}{g \choose m}{f \choose \beta}\sum_{i=0}^{\beta}(-1)^i\binom{\beta}{i}\binom{i\cdot s_0}{d-m}
        \]
which --- combined with \eqref{eq: R_beta, t_beta} --- gives
\[|\mathbf{R}_{\beta,\mathcal{I}}|=\sum_{m=1}^{\min\{d-\beta, g\}}{g \choose m}\sum_{i=0}^{\beta}(-1)^i\binom{\beta}{i}\binom{i\cdot s_0}{d-m}\] which concludes the proof.

\subsection{Constructing \(\mathbf{R}_{\beta, \mathcal I, \sigma}\)}
\label{appendix: R_beta,I}
We here follow the path of the design of $\mathbf{C}_{\beta,\mathcal I,\sigma}$ from Appendix~\ref{appendix: C_beta,I}, and we again form sets indexed by \(\beta\) and \(\mathcal{I}\); however, in this case \(\beta\) ranges over \([\beta_{\mathrm{min}}, \beta_{\mathrm{max}}]\), and we simply use \(\mathbf{R}_{\beta,\mathcal I}\) instead of \(\mathbf{C}_{\beta,\mathcal I}\). This same approach yields the allocated set $\mathbf{R}_{\beta,\mathcal I,\sigma}$ of \(d\)-tuples to group \(\sigma\) corresponding to this \(\mathcal I\), which takes the form
\begin{equation}\label{eq:allocated-set,R}
\mathbf{R}_{\beta,\mathcal I,\sigma}
=
\Big\{
  \mathbf{r}_{\beta,\mathcal I}^{(\ell)}
  \mid 
  \ell = s_{j,\beta,\mathcal I},\,
        s_{j,\beta,\mathcal I}+1,\,
        \dots,\,
        e_{j,\beta,\mathcal I}
\Big\}.
\end{equation}

\subsection{Bounding the Range of \(N\)}

\label{Appendix: limit on N}

Recall from~\eqref{eq:largest_r} that 
$k = \max \left\{ r \in \mathbb{Z}^{+} \,\mid\, \binom{r}{d} \leq N \right\}$, and let us consider the case where $ k\nmid n$. Also recall from \eqref{eq:s_0} that \( s_0 = \left\lfloor \frac{n}{k + d} \right\rfloor + 1 \), and let us define
\begin{equation} \label{Kvalid1}
        \mathcal{K}_{\mathrm{valid}} \triangleq \left\{k \in [d,n] \subset \mathbb{Z} \ \mid  \ \left\lfloor\frac{n}{k + d}\right\rfloor + 1 \leq \left\lfloor\frac{n}{k} \right \rfloor   \ \right\}.
\end{equation}
We now claim that if \(k \in \mathcal{K}_{\mathrm{valid}}\), then for some \( 0 \leq g \leq s_0 \cdot d \), the number of files \( n \) can be written as
\begin{equation}
\label{eq:DDQR formula}
    n=k\cdot s_0+g.
\end{equation}
To prove the claim, suppose \( k \in \mathcal{K}_{\mathrm{valid}} \), which would mean that
\begin{equation}
   \label{eq: appendix, kvalid,1} 
   s_0 \leq \left\lfloor \frac{n}{k} \right\rfloor
\end{equation}
which also means that 
\begin{equation}
\label{eq: appendix, kvalid,2}
  g = n - s_0 \cdot k \geq n - \left\lfloor \frac{n}{k} \right\rfloor \cdot k \geq 0
\end{equation}
and thus we have 
\begin{equation}
\label{eq: appendix, kvalid,3}
    g = n - s_0 \cdot k 
= n - \left( \left\lfloor \frac{n}{k + d} \right\rfloor + 1 \right) \cdot k 
\leq n - \left\lceil \frac{n}{k + d} \right\rceil \cdot k 
\leq n - \frac{n}{k + d} \cdot k 
= \frac{n \cdot d}{k + d}.
\end{equation}

Since \( \frac{n}{k + d} < \left\lfloor \frac{n}{k + d} \right\rfloor + 1 = s_0 \), we can now conclude that
\begin{equation}
    \label{eq: appendix, kvalid,4}
    g \leq \frac{n \cdot d}{k + d} \leq s_0 \cdot d.
\end{equation}
Let us now translate this to a constraint on $N$. To ensure that the condition \eqref{eq: appendix, kvalid,1} holds --- and accounting for the fact that when two real values differ by at least one, their respective floor values must also differ by at least one --- it is now sufficient to guarantee that
\begin{equation}
\label{eq: appendix, kvalid,5}
 \frac{n}{k} - \frac{n}{k + d} \geq 1
\end{equation}
i.e., to guarantee that
\begin{equation}
\label{eq: appendix, kvalid,6}
   \frac{n d}{k(k + d)} \geq 1
\end{equation}
which in turn yields the quadratic inequality condition
\begin{equation}
\label{eq: appendix, kvalid,7}
    k^2 + dk - dn \leq 0.
\end{equation}
Solving this, yields
\begin{equation}
\label{eq: appendix, kvalid,8}
k \leq \frac{-d + \sqrt{d^2 + 4dn}}{2} =: k_{\mathrm{max}}
\end{equation}
and we can now conclude that 
\begin{equation}
  N \leq N_{\mathrm{max}}\triangleq {k_{\mathrm{max}}  \choose d}.
\end{equation}
\textcolor{black}{It is not difficult to see that for \(n\ge 32 d\), we have \( \frac{9}{10}\sqrt{dn}\le k_{\mathrm{max}} \leq \sqrt{dn}\),} which, after applying Sterling's property, yields   
\begin{equation}
   N_{\mathrm{max}} \ge\left( \frac{9}{10}\sqrt{\frac{n}{d}} \right)^d.
\end{equation}
We use \(\left( \frac{9}{10}\sqrt{\frac{n}{d}} \right)^d\) throughout the paper to bound the range of acceptable $N$.

\end{appendices}

\printbibliography

\end{document}